%% file: Submission.tex
\documentclass[a4paper,UKenglish]{lipics-v2016}
 
\usepackage{microtype}
\usepackage{algorithm}
\usepackage{algpseudocode}
\graphicspath{{./graphics/}}

\input{input.tex}

\title{Submodular Maximization Under A Matroid Constraint: Asking more from an old friend, the Greedy Algorithm} 
  \titlerunning{Submodular Maximization Under A Matroid Constraint}
\author[1]{Nived Rajaraman}
\author[2]{Rahul Vaze}
\affil[1]{%
  Indian Institute of Technology, Madras, \\ Department of Electrical Engineering, \\Indian Institute of Technology, Madras, Chennai, India
\texttt{nived.rajaraman@gmail.com}}

\affil[2]{%
  Tata Institute of Fundamental Research, Mumbai, \\
  School of Technology and Computer Science, \\ Tata Institute of Fundamental Research
  Mumbai, India,
\texttt{vaze@tcs.tifr.res.in}}
\begin{document}
\maketitle


\begin{abstract}The classical problem of maximizing a submodular function under a matroid constraint is considered. Defining a new measure for the increments made by the greedy algorithm at each step, called the {\bf discriminant}, improved approximation ratio guarantees are derived for the greedy algorithm. At each step, discriminant measures the multiplicative gap in the incremental valuation between the item chosen by the greedy algorithm and the largest potential incremental valuation for eligible items not selected by it.
The new guarantee subsumes all the previous known results for the greedy algorithm, including the curvature based ones, and the derived guarantees are shown to be tight via constructing specific instances. More refined approximation guarantee is derived for a special case called the submodular welfare maximization/partition problem that is also tight, for both the offline and the online case.
\end{abstract}

%
%



\input{Vazeintro.tex}

\input{MainFile.tex}

\nocite{*}
\bibliography{refs}

\end{document}

%% file: input.tex
\usepackage{amsfonts}
\usepackage{times}
\usepackage{latexsym}
\usepackage{amssymb}
\usepackage{amsmath}
\usepackage{natbib}
\usepackage{verbatim}

\newtheorem{bemark}[theorem]{Remark}
\newenvironment{proofsk}{ \textbf{Proof Sketch:} }{ \hfill $\Box$}

\def\bb0{{\mathbb{0}}}

\def\bb{{\mathbf{b}}}

\def\b0{{\mathbf{0}}}

\def\b1{{\mathbf{1}}}

\def\bbR{{\mathbb{R}}}

\def\cM{\mathcal{M}}

\def\cR{\mathcal{R}}

\def\sfr{{\mathsf{r}}}

\def\sf0{{\mathsf{0}}}


%% file: Vazeintro.tex
\section{Introduction}

We consider the problem of maximizing a submodular function under a matroid constraint. This is a classical problem (\cite{edmonds1971matroids}), with many important special cases, e.g., uniform matroid (the subset selection problem), partition matroid (submodular welfare/partition problem). The problem is known to be NP-hard even for special cases, and the earliest theoretical results on this problem date back to the seminal work of  \cite{edmonds1971matroids, nemhauser1978best, Fisher1978}, that derived tight approximation guarantees. 
In particular, for the general problem, the greedy algorithm is known to achieve a $1/2$-approximation (\cite{Fisher1978}), while for the uniform matroid (subset selection problem), where the objective is to select the optimal subset under a cardinality constraint, the greedy algorithm is known to be $(1-1/e)$-approximate (\cite{edmonds1971matroids, nemhauser1978best}). Moreover, these guarantees are tight and instance independent. 

Even though the theoretical work limits the instance independent guarantees for the greedy algorithm to be $1/2$ or $1-1/e$, in practice, the performance of the greedy algorithm is far better, sometimes even close to the optimal. 
To explain this phenomenon, work on instance dependent guarantees started with (\cite{CONFORTI1984251}), which showed that using the concept of {\it curvature}, the approximation guarantee of the greedy algorithm can be improved from $1/2$ to $\frac{1}{1+c}$, where $c$ is the curvature that captures the distance of the function from being linear. Lower curvature is better, with zero-curvature (modular) giving the optimal solution. For the special case of subset selection problem, the guarantee can be improved from $(1-1/e)$ to 
$(1-e^{-c})/c$ ($\rightarrow 1$ as $c \rightarrow 0$) (\cite{CONFORTI1984251}). 
In more recent work,  (\cite{DBLP:journals/corr/abs-1709-02910}) has improved the guarantee to $(1-\gamma_h/e-\epsilon)$, where $\gamma_h$ is the $h$-curvature and for any $\epsilon>0$. Most of the theoretical work on the greedy algorithm assumes the value oracle model, where a polynomial algorithm is assumed to exist that can compute the optimal increment in each iteration. To obviate this possibly restrictive assumption, approximate greedy algorithms were considered by \cite{goundan2007revisiting}, where the increment is only available up to a certain approximation guarantee. 
 There is also work on finding guarantees for non-monotone submodular maximization problems (see, e.g., \cite{feige2011maximizing}). 

Compared to deterministic algorithms, randomized algorithms (\cite{Calinescu:2011:MMS:2340436.2340447}) can improve the instance dependent guarantee to $(1-e^{-c})/c$ (with the continuous greedy algorithm) for the general problem, which can be further refined for the subset selection problem to 
$(1-c/e)$ (\cite{DBLP:journals/corr/SviridenkoW13}) using a non-oblivious local search algorithm.

In addition to the subset selection problem, another important special case of the general problem is the submodular welfare maximization/partition problem, where there is a set of items/resources $\cR$ that has to be partitioned among the set of $n$ agents, each agent has a submodular valuation function $f_i$ over the subsets of $\cR$, and the problem is to find the partition of $\cR$ that maximizes the sum of the agents' valuations after partition. This problem was addressed in (\cite{Fisher1978}) itself that gave a $1/2$-approximate guarantee for the greedy algorithm, which surprisingly holds even in the online setting (the elements of the resource set $\cR$ are revealed sequentially, and on arrival of each new element it has to be assigned irrevocably to one of the agents). Randomized algorithm with 
$(1-1/e)$-approximation guarantee was proposed in (\cite{vondrak2008optimal}) for this problem. Instance dependent guarantees as a function of the curvature $c$ for the general problem of course carry over to this problem as well. 

For deriving instance dependent guarantees, the motivation to consider the curvature of the submodular function was that if the curvature is small, then the greedy algorithm remains `close' to the optimal solution. In a similar spirit, in this paper, we consider a 
new measure of the problem instance and the greedy algorithm, called the {\it discriminant}, where larger discriminant helps the greedy algorithm to stay `close' to the optimal. 
We exploit the discriminant for finding improved instance dependent guarantees for the greedy algorithm when used to solve the general problem, the partition problem, and the online partition problem under the value oracle model. 

We begin the discussion on discriminant using the partition problem and then describe the corresponding definition of the discriminant for the general problem.
For the partition problem, the greedy algorithm at each iteration choses the item-agent pair that maximizes the incremental valuation, and assigns the chosen item to the chosen agent.
We define the discriminant $d_s$ at iteration $s$ of the greedy algorithm, as the ratio of the incremental increase in valuation made by the greedy algorithm (because of the item-agent pair chosen by greedy) and the best incremental increase (possible) in valuation among all other agents (not chosen by greedy) for the item chosen by the greedy algorithm in iteration $s$, given the past choices of the greedy algorithm until iteration $s-1$.
Formal definition of $d_s$ is given in Definition \ref{defn:disc}. It is easy to see that uniformly (over all iterations) large discriminant should help the greedy algorithm in staying `close' to the optimal solution.

The intuition behind considering the discriminant becomes clear especially for the following  asymmetric partition problem, where among the $n$-agents, one of them (say $i$) has a valuation such that $f_i(S) > > f_j(S)$  $\forall \ j \ne i, \forall \ S \subset \cR$. Clearly, the greedy algorithm (assign all resources to agent $i$) is optimal, while
the best known bound for it over this instance is
$(1-e^{-c})/c$, where $c= \max_i c_i$ ($c_i$ is the curvature for user $i$) can be large, resulting in a poor guarantee.
Incidentally though, the discriminant remains uniformly large throughout the execution of the greedy algorithm for this example, indicating that discriminant may be related with the performance of the greedy algorithm and can imply better guarantees. 


So the first question we ask: can we generalize this intuition and derive a theoretical guarantee on the approximation ratio of the greedy algorithm for the partition problem as a function of the discriminant without losing out on the dependence of the curvature. The answer turns out to be positive: we show that the greedy algorithm (with a slight modification for tie-breaking) can achieve an approximation ratio of 
\begin{equation}\label{eq:res1}
\min\left\{1, \frac{1}{\max_s\{\frac{1}{d_s} + c_s\}}\right\},
\end{equation} where $c_s$ is the curvature of the user chosen at iteration $s$ of the greedy algorithm. 
This result nicely explains the optimal performance of the greedy algorithm for the asymmetric partition problem, since for that $d_s$ is very large for all $s$, and our approximation guarantee approaches $1$ as $d_s \rightarrow \infty, \ \forall \ s$ and since $c_s\le 1$. By definition $d_s\ge 1$, thus, compared to the previous best known approximation guarantee of $\frac{1}{1 + c}$ ($c = \max_{u \in\text{user set}} c_u$) for the greedy algorithm (\cite{CONFORTI1984251}), our result is stronger unless $d_s=1$ for some iteration $s$ and the curvature of the user chosen in iteration $s$ is the largest among all the users, in which case it equals $\frac{1}{1 + c}$. So our result provides a newer and stronger guarantee for asymmetric problems, when $d_s$ remains large, and exploits a new dimension (discriminant) of the submodular partition problem that is tied to the greedy algorithm.  
Similar to the earlier approximation guaratees for the greedy algorithm (\cite{Fisher1978}), we show that the derived approximation guarantee \eqref{eq:res1} holds  for even the online partition problem, and the bound is tight. We refer the reader to (\cite{44224}) for more details and review of the recent progress on the online problem under stochastic/secretarial settings). We also provide some intuition on the specific form of dependence of the discriminant on the approximation guarantee \eqref{eq:res1} in Remark \ref{rem:intuition}.


Next, the natural question is: can we derive discriminant dependent guarantees for the greedy algorithm for the general submodular maximization problem under a matroid constraint. The answer to this question is also yes, however, the guarantee is little different and is given by 
\begin{equation}\label{eq:result}\min\left\{1, \frac{1}{\{\frac{1}{d_{\min}} + c\}}\right\},
\end{equation} where  $d_{\min} = \min_s d_s$, and the discriminant at iteration $s$ $d_s$ for this case is defined as: given the past choices of the greedy algorithm until iteration $s-1$, $d_s$ is the ratio of the incremental increase in valuation (item chosen by greedy) and the best incremental increase (possible) in valuation among all other items (other than the one chosen by the greedy that are still available for selection) in iteration $s$ (formal definition is provided in Definition \ref{defn:discgen}). Once again as $d_{\min}\ge 1$, our guarantee subsumes the previous known result of $1/(1+c)$, and matches that only if $d_s=1$ is some iteration $s$. 

The intuition for the approximation guarantee \eqref{eq:result} can be developed by considering the special case of the uniform matroid (subset selection problem), where the greedy algorithm selects a new element that has the largest incremental increase 
in function valuation at each iteration. 
If in each iteration $s\ge 1$, $d_s$ is large, the greedy algorithm is making rapid progress towards the optimal valuation, by selecting `near-optimal' elements, since the elements it rejects have comparatively low incremental valuation. We also show that this guarantee is tight for the general problem. 

Our work exploits an unexplored parameter of the greedy algorithm, discriminant, 
and provides a new guarantee that subsumes all previous guarantees. 
The utility of discriminant based guarantee is easily manifested for the submodular partitioning problems, where the valuation functions for different users have inherent asymmetry, such that the discriminants are uniformly large, e.g. in subcarrier and power allocation in wireless systems (\cite{thekumparampil2016combinatorial}).  
As far as we know, this is the first time an algorithm dependent (greedy algorithm) and instance dependent approximation guarantee has been derived for submodular maximization problem. Even though the guarantee (discriminant) is algorithm dependent, however, since the greedy algorithm is deterministic, the discriminants can be computed once the problem instance is specified, and computational complexity of finding the discriminant is same as the complexity of the greedy algorithm. 

In this paper, we have discovered a new connection between the multiplicative gap (which we call the discriminant) between the locally best increment made by the greedy algorithm and the next best increment possible, and its performance guarantee for the submodular maximization problem that has wide applications. The discriminant appears to be a fundamental quantity in studying greedy algorithms, and we believe that such an approach can also lead to improved guarantees for similar combinatorial problems, e.g., the generalized assignment (GAP) problem (\cite{fleischer2006tight}), where the greedy algorithm is known/observed to perform well.

%% file: MainFile.tex
\newtheorem{problem}{Problem}
\section{Monotone Submodular Maximization over a Matroid}

\begin{definition} \label{d1}
(Matroid)
\noindent 
A matroid over a finite ground set $N$ is a pair $(N,\mathcal{M})$, where $\mathcal{M}\subseteq 2^N$ (power set of $N$) that satisfies the following properties:
\begin{enumerate}
    \item $\phi \in \mathcal{M}$,
    \item If $T \in \mathcal{M} \text{ and } S \subset T \Rightarrow S \in \mathcal{M}$ [independence system property],
    \item If $S, T \in \mathcal{M} : |T| > |S| \Rightarrow \exists \ x \in T \setminus S : S \cup \{ x\} \in \mathcal{M}$ [augmentation property].
    $\mathcal{M}$ is called the family of independent sets.
\end{enumerate}
\end{definition}

\begin{definition} \label{d2}
(Rank of a matroid) For $S\subseteq N$, the rank function of a matroid $(N,\mathcal{M})$ is defined as 
$r(S) = \max\{|M|: M\subseteq S, M \in \cM\}$, and rank of the matroid is $r(N)$, the cardinality of the largest independent set. For a matroid to have rank $K$, there must exist no independent sets of cardinality $K+1$.
\end{definition}


For many applications, two special cases of matroids are of interest, namely the uniform and the partition matroid that are defined as follows.

\begin{definition} \label{d4}
(Uniform Matroid)
\noindent For some $K>0$, the uniform matroid over a ground set $N$ is defined as $(N,\mathcal{M}^u)$, where $\mathcal{M}^u = \{ S : S \subseteq N, |S| \le K\}$.
The uniform matroid has rank $K$.
\end{definition}

\begin{definition} \label{d5}
(Partition Matroid) A ground set $N$, and its partition $\{ P_i : i = 1,2,\dots,p\}$, 
$\cup_i P_i = N$, and $\ P_i \cap P_j = \phi, i \ne j$ are given. Given integers 
$k_i : 1 \le k_i \le |P_i|$, the partition matroid over $N$ is defined as 
$(N,\mathcal{M}^p)$, where $\mathcal{M}^p = \{ S : S \subseteq N \text{ and } |S \cap P_i| \le k_i \text{ for } i = 1,2,\dots,p \}$.
This partition matroid has rank $\sum_{i=1}^p k_i$.\end{definition}

\begin{definition} \label{defn:submod}
(Monotone and Submodular Function)
A set function $Z: 2^N \rightarrow \mathbb{R}$ is defined to be
monotone if for $S\subset T \subseteq N$, $Z(S) \le Z(T)$,  and submodular if for all 
$ T \subseteq N \setminus \phi, S \subset T \text{ and } x \notin T$, 
$\ Z(T \cup \{x\}) - Z(T) \le Z(S \cup \{x\}) - Z(S) \text{ [Diminishing Returns]}$.
Without loss of generality, we assume $Z(\phi)=0$ $(\Rightarrow Z(S) \ge 0, \forall \ S \subset N)$.
\end{definition}
\begin{problem} \label{prob:1}
Given a matroid $(N,\cM)$ of rank $K$, and a monotone and submodular function $Z:2^N\rightarrow \bbR$, the problem is to find $ \text{max }\{Z(S) : S \in \mathcal{M}\}$.
\end{problem}


Since $Z(S)$ is non-decreasing, 
we only consider feasible solutions to Problem \ref{prob:1} that have cardinality $K$ even if there exists a smaller solution with the same valuation. For the rest of the paper we use the following notation. Given a set $S$ associated with an ordering $(s_1,s_2, \dots, s_{|S|})$, $S^{i}$ denotes the partial ordering $(s_1,s_2, \dots, s_{i})$.
 The increment in valuation of set $S$ upon addition of element $q$ to $S$ is defined as
$$ \rho_q(S) = Z(S \cup \{q\}) - Z(S). $$ 
The most natural algorithm to solve Problem \ref{prob:1} is a greedy algorithm (Algorithm \ref{alg0}) that incrementally adds an element to the existing set that provides the largest increase in the set valuation as described next.
\begin{algorithm}[H]
\caption{Greedy Algorithm for monotone submodular maximization over a rank-$K$ matroid $(N,\mathcal{M})$} \label{alg0}
\begin{algorithmic}
\Procedure{$\mathsf{GREEDY}$}{}
    \State \textbf{Initialize:} $G^0=\phi, i=1$
    \While{$i \le K$}
        \State $q_i \gets \underset{q}{\text{ argmax}}\left\{ \rho_q(G^{i-1}) : \ q \cup G^{i-1} \in \mathcal{M} \right\}$
        \Comment{Pick arbitrarily in case of ties}
        \State $G^i \gets G^{i-1} \cup \{q_i\}$
        \State $i \gets i+1$
    \EndWhile 
    \State \textbf{Return} $G=G^K$
\EndProcedure
\end{algorithmic}
\end{algorithm}

We let $\Omega$ denote the global optimal solution and $G$ the solution generated by the $\mathsf{GREEDY}$ algorithm for the problem in context, respectively.
Then the following guarantees are known for Problem \ref{prob:1}.
\begin{theorem} \cite{Fisher1978} \label{t1}For Problem \ref{prob:1}, $\frac{Z(G)}{Z(\Omega)} \ge \frac{1}{2}$.
\end{theorem}

\noindent The above $\frac{1}{2}$-approximation bound is a global lower bound on the performance of the greedy algorithm for Problem \ref{prob:1}, which can be further improved with the knowledge of the \textbf{curvature} parameter, $c$ of the monotone and submodular valuation function, that is defined as
\begin{equation}
c = 1 - \underset{S,j\in S^*}{\text{min }} \frac{\rho_j(S)}{\rho_j(\phi)}, \quad \text{where } \ S^* = \{j: j \in N \setminus S,\ \rho_j(\phi) > 0\}. \label{def:curvature}
\end{equation}

\noindent By submodularity, we have $c \le 1$. The case $c = 0$ implies that the function valuations are linear.

\begin{theorem} [\cite{CONFORTI1984251} Theorem 2.3] \label{t2}
For Problem \ref{prob:1}, $\frac{Z(G)}{Z(\Omega)} \ge \frac{1}{1+c}$.
\end{theorem}

The first main result of this paper is presented in the Theorem \ref{t6}, where we derive stronger approximation guarantees for the $\mathsf{GREEDY}$ algorithm as a function of `discriminant' (as defined next) in addition to the curvature as done in Theorem \ref{t2}. 
\begin{definition} \label{defn:discgen}
Let $S_\bot := \{ j : j \in N \setminus S,\ j \cup S \in \mathcal{M}\}$. Then, given the selected set $G^{i-1}$ by the $\mathsf{GREEDY}$ algorithm at the end of iteration $i-1$, the \textbf{discriminant}  $d_i$  at iteration $i$ is defined as:
$$d_i=  \frac{ \rho_{g_i}(G^{i-1}) }{\underset{ g'_i \ne g_i}{\text{max }} \rho_{g_i'}(G^{i-1}) },$$
where $g_i = \underset{g}{\text{argmax }}\rho_g(G^{i-1}), \text{ and } g'_i \in G^{i-1}_\bot. $
Define $d_i := \infty \ \text{ if } \forall g'_i \in G^{i-1}_\bot : g'_i \ne g_i,\ \rho_{g_i'}(G^{i-1}) = 0$.
Moreover, the \textbf{minimum discriminant} is defined as
$d_{\text{min}} = \underset{i < i_0}{\text{min }} d_i$, where $i_0 = \text{min }\{i: |G^{i-1}_\bot| = K-i+1 \}$, where $K$ is the rank of the matroid.
\end{definition} 
Consider the element $g_i$ selected by the $\mathsf{GREEDY}$ algorithm in iteration $i$ and consider the incremental gain $\rho_{g_i}(G^{i-1})$. Find the best element $g_i'$ other than $g_i$ such that $G^{i-1} \cup g_i' \in \mathcal{M}$, and compute the incremental increment $\rho_{g_i'}(G^{i-1})$. The ratio of the two is defined as the discriminant in iteration $i$. Moreover, index $i_0$ is the earliest iteration $i$ in the execution of the $\mathsf{GREEDY}$ algorithm where the number of items that are not part of the $\mathsf{GREEDY}$ algorithm's chosen set equals $K-i+1$. This is useful, since after iteration $i_0$ we show that the items chosen by the $\mathsf{GREEDY}$ algorithm belong to the optimal set as well. Thus, to compute $d_{\text{min}}$, minimization needs to be carried out over a smaller set.

\begin{theorem} \label{t6}
    Using the $\mathsf{GREEDY}$ algorithm (Algorithm \ref{alg0}) for Problem \ref{prob:1} guarantees:
    $$\frac{Z(G)}{Z(\Omega)} \ge \text{min} \left( 1, \frac{1}{ \left(c + \underset{i < i_0}{\text{max }}\frac{1}{d_i}\right)}\right) = \text{min} \left( 1, \frac{1}{c + \frac{1}{d_{\text{min}}}}\right).$$
\end{theorem}

{\it Discussion:} Since $d_{\text{min}}\ge 1$, the approximation guarantee provided by Theorem \ref{t6} subsumes the best known guarantee for the greedy algorithm (Theorem \ref{t2}), 
and matches that only when there is a tie in some iteration before $i_0$, in which case it matches the result of Theorem \ref{t2} ($\frac{1}{1+c}$). Theorem \ref{t6} shows that if the problem instance has large discriminants, the greedy algorithm is theoretically far better than what was previously known. The proof of Theorem \ref{t6} is rather technical and does not allow simple intuitive explanation. 
We provide intuition for the specific form of the guarantee as the function of the discriminant ($1/d$) for the special case of the submodular partition problem in the next section in Remark \ref{rem:intuition}.

All proofs are provided in the appendices. Here we give a brief proof sketch for Theorem \ref{t6}.

\begin{proofsk}The key step in proving Theorem \ref{t6} is to show that 
\begin{equation}\label{eq:sketch} Z(\Omega)\le c\sum_{i : g_i \in G \setminus \Omega}\rho_{g_i}(G^{i-1}) + \sum_{i : g_i \in G \cap \Omega}\rho_{g_i}(G^{i-1}) + \sum_{i : \omega_i \in \Omega \setminus G} \rho_{\omega_i}(G),
\end{equation} via Lemma \ref{l5}, and using a particular ordering for $\Omega$ that is a function of the ordering of $G$ (defined in Lemma \ref{l2}). We then show in Lemma \ref{l9} that for each $i$, $\rho_{\omega_i}(G^{i-1})$ (which upper bounds $\rho_{\omega_i}(G)$) cannot be larger than $\frac{\rho_{g_i}(G^{i-1})}{d_i}$, and substitute this back into the upper bound \eqref{eq:sketch} for $Z(\Omega)$. Subsequently, as shown in Section \ref{sec6p1} B, the set $\{i : i \ge i_0\}$ is a subset of $\{i : g_i \in G \cap \Omega\}$, which implies that $\{i : i < i_0\}$ is a superset of $\{i : g_i \in G \backslash \Omega\}$ as well as $\{i : g_i \in \Omega  \backslash G\}$, and we  replace the summation index set for the first and the third term in \eqref{eq:sketch} with $\{i : i < i_0\}$
to further upper bound  $Z(\Omega)$. The statement of Theorem \ref{t6} then follows by rearranging terms and simplifying.
\end{proofsk}


\subsection{Discussion on the definition of discriminant for Problem \ref{prob:1}} \label{sec6p1}
Essentially, the discriminant in each iteration (Definition \ref{defn:discgen}) is the ratio of the locally best increment made by the $\mathsf{GREEDY}$ algorithm and the next best possible increment at any iteration in its execution. However, there is a little subtlety which is explained as follows.  Note the following cases:
\begin{enumerate}
    \item[A] If $\underset{ g'_i \ne g_i }{\text{max }} \rho_{g_i'}(G^{i-1}) = 0$, $d_i$ is defined as $\infty$ for the sake of continuity and can be removed from the minimization over $i$ in the definition of $d_{\min}$. If $i_0 = 1$, the problem is trivial (there exists only $1$ feasible solution).
    \item[B] In the following argument, we show that in every iteration after $i_0$, the element chosen by the $\mathsf{GREEDY}$ algorithm belongs to the optimal solution, i.e., $\{i : i \ge i_0\} \subseteq \{i : g_i \in G \cap \Omega\}$. Let at a iteration $i$, $|G^{i-1}_\bot| = K-i+1$, that is, there exist exactly $K-i+1$ valid choices of elements that can be added to $G^{i-1}$ to generate $G^i$.
    \begin{enumerate}
        \item For such $i$, Lemma \ref{lemma:choices} states that $S=G^{i-1}_\bot$ is a unique set that satisfies $|S| = K-i+1$ and $G \cup G^{i-1}_\bot \in \mathcal{M}$. Thus it follows from Corollary \ref{c2}, that $G^{i-1}_\bot \subseteq \Omega$, where $\Omega$ is the optimal solution.
        \item From the uniqueness of $S$, it follows that $G = G^{i-1} \cup G^{i-1}_\bot$. Thus, $G^{i-1}_\bot \subseteq G$.
        \item From (a) and (b), it follows that for iteration $i$, $|G^{i-1}_\bot| = K-i+1$, then $G^{i-1}_\bot \in G \cap \Omega$.
        \item If in iteration $i$, $|G^{i-1}_\bot| = K-i+1$, then for all $j>i$, $|G^{j-1}_\bot| = K-j+1$.
    \end{enumerate}
    Recalling the definition of $i_0 = \text{min }\{i: |G^{i-1}_\bot| = K-i+1 \}$, from (c) and (d) we conclude that for every iteration $i \ge i_0$, we have $
            G^{i-1}_\bot \in G \cap \Omega$.

\end{enumerate}

\begin{lemma}\label{lem:submodtight} The approximation guarantee obtained in Theorem \ref{t6} is tight.
\end{lemma}
The result of Theorem \ref{t6} is derived by upper bounding \eqref{eq22}. 
To prove Lemma \ref{lem:submodtight}, in Appendix \ref{app:tightgen}, we provide a problem instance for which the approximation guarantee matches the bound derived in \eqref{eq22}.

Next, we consider the special case of Problem \ref{prob:1}, the submodular partition problem, and provide better guarantees than the general problem as a function of the discriminants.

\section{Submodular Partition Problem}

\begin{problem} \label{p2}
Given a set of allocable resources $\mathcal{R}$ with $|\mathcal{R}|=n$, and a set of users denoted by $\mathcal{U}$ with $|\mathcal{U}| = m$. Each user $u$ has a monotone and submodular valuation function $Z_u(S) : 2^\mathcal{R} \rightarrow \mathbb{R}$, where without loss of generality, $Z_u(\phi) = 0, \forall u \in \mathcal{U}$. 
The submodular partition problem is to find a partition of the set of resources $\mathcal{R}$, among the set of users $\mathcal{U}$ such that the sum of the valuations of individual users is maximized. That is:
$$\text{max } \sum_{u \in \mathcal{U}} Z_u(S_u), \text{subject to:} \ S_u \subseteq \mathcal{R} \quad \forall u, \ S_{u_i} \cap S_{u_j} = \phi, \text{for }  u_i \ne u_j.$$
\end{problem}

\noindent The submodular partition problem is a special case of Problem \ref{prob:1}, where the matroid is the partition matroid $(\cR, \cM^p)$,
$$\mathcal{M}^p = \{ S: S \subseteq \mathcal{V}, \ |S \cap \mathcal{V}_r| \le 1 \quad \forall r \in \mathcal{R} \},$$
where $\mathcal{V}=\mathcal{U} \times \mathcal{R} = \{ (u,r) : u \in \mathcal{U}, r \in \mathcal{R}\}$, and $\mathcal{V}_{r}=\{ (u,r) : u \in \mathcal{U}\}$ and  $Z(S)=
\sum_{u \in \mathcal{U}} Z_u(S_u)$ is submodular, since the sum of submodular functions is submodular. For this special case, we denote the increment in valuation by allocating resource $r$ to user $u$ given the existing set $S$ as:
\begin{equation} \label{eq:rhodef}
\rho_r^u(S) := Z(S \cup \{(u,r)\})-Z(S) = Z_u(S_u \cup \{r\}) - Z_u(S_u).
\end{equation}
\noindent Since Problem \ref{p2} is a special case of Problem \ref{prob:1}, Theorem \ref{t2} implies an approximation guarantee of $\frac{1}{1+c}$, where $c$ is the curvature of $\sum_{u \in \mathcal{U}}Z_u(S_u)$. It is known that $c = \underset{u \in \mathcal{U}}{\text{max }} c_u$.
%

Next, we describe a modified greedy algorithm, called $\mathsf{GREEDY-M}$, where the modification is in the tie-breaking rule compared to $\mathsf{GREEDY}$,  that uses curvature and incremental gain ratios, and derive improved approximation guarantee in  Theorem \ref{t6} for it compared to Theorem \ref{t2}. 
To define the modified greedy algorithm, we need the following definition.
Let for  $S \in \mathcal{M}^p$, the set of unallocated resources in $S$ be defined as 
$$ \mathcal{R}(S_\bot) := \{ r : \forall u \in \mathcal{U}, (u,r) \not\in S \}, $$
i.e., the resources that do not appear in the set of user-resource pairs that are part of $S$.

\begin{algorithm}[H]
\caption{Modified Greedy Algorithm for the Submodular Partition problem} \label{alg1}
\begin{algorithmic}[1]
\Procedure{$\mathsf{GREEDY-M}$}{}
    \State \textbf{Initialize:} $G^0=\phi, i=1$ 
    \While{$i \le |\cR|=n$} 
        \State $(u^*,r^*) \gets \underset{u,r}{\text{ argmax}}\left\{ \rho_r^u(G^{i-1}) : \ r \in \mathcal{R}(G^{i-1}_\bot) \right\}$\\
        \Comment{\textbf{Tie Breaking Rule}: In case there is more than one optimal pair, choose the user-resource pair $(u,r)$ that minimizes $c_u + \frac{1}{d_i(u,r)}$}, where $d_i(u,r) := \frac{\rho_{r}^{u}(G^{i-1})}{\max_{u' \ne u } \rho_{r}^{u'}(G^{i-1})}$
        \State $G^i \gets G^{i-1} \cup \{(u^*,r^*)\}$
        \State $i \gets i+1$
    \EndWhile
    \State \textbf{Return} $G=G^n$
\EndProcedure
\end{algorithmic}
\end{algorithm}

\begin{definition} \label{defn:disc}
\noindent Let the set chosen by the $\mathsf{GREEDY-M}$ algorithm at the end of iteration $i-1$ be $G^{i-1}$. Consider the user-resource pair $(u^*,r^*)$ selected by the greedy algorithm in iteration $i$ and consider its incremental gain $\rho_{r^*}^{u^*}(G^{i-1})$. 
For the resource $r^*$ chosen by the $\mathsf{GREEDY-M}$, find the best user $u'$ other than $u^*$ with the highest incremental valuation for $r^*$ and compute the incremental gain $\rho_{r^*}^{u'}(G^{i-1})$. The ratio of the two incremental gains is defined as the discriminant $d_i^p$ ($p$ stands for the partition problem) in iteration $i$, as 
$$d_i^p = \frac{\rho_{r^*}^{u^*}(G^{i-1}) }{\underset{ u' \ne u^*}{\text{max }} \rho_{r^*}^{u'} (G^{i-1}) }, \quad \text{where } (u^*,r^*) = \underset{u,r \in \mathcal{R}(G^{i-1}_\bot)}{\text{argmax }} \rho_r^{u}(G^{i-1}). $$
With respect to the $\mathsf{GREEDY-M}$ algorithm, $d_i^p = d_i(u^*,r^*)$.
\end{definition} 
At iteration $i$, discriminant $d_i^p$ is the ratio of the increment due to the best local user-resource pair chosen by the $\mathsf{GREEDY-M}$ algorithm and the increment possible if the resource chosen by the $\mathsf{GREEDY-M}$ algorithm is allocated to the user who values it second most.

\begin{bemark} For $\mathsf{GREEDY-M}$ algorithm, if the tie-breaking rule also fails to produce a unique pair, in which case a particular indexing of user resource pairs that is fixed at the beginning is used,  and the user resource pair with the highest index is declared the chosen pair in that iteration. Thus, the value of discriminant $d_i^p$ in iteration $i$ is uniquely defined once the problem is specified.
\end{bemark}


Next, we present the second main result of this paper, that gives an approximation guarantee for the $\mathsf{GREEDY-M}$ algorithm as a function of the discriminant and the curvature.

\begin{theorem} \label{t5}
    Using the $\mathsf{GREEDY-M}$ algorithm for Problem \ref{p2} guarantees:
        $$\frac{Z(G)}{Z(\Omega)} \ge \text{min} \left( 1, \frac{1}{ \underset{i}{\text{max}}\left\{c_{u_i} + \frac{1}{d_i^p}\right\}}\right),$$
 where $d_i^p$ is the discriminant and $c_{u_i}$ is the curvature of the user chosen, in iteration $i$, respectively, for $i= \{1,2,\dots, |\cR|\}$.      
\end{theorem}

{\bf Discussion:}
Compared to Theorem \ref{t2}, $1/d_i^p$ replaces the $1$ in the denominator, and thus larger the discriminant better is the approximation guarantee, and which approaches $1$ as  $d_i^p\rightarrow \infty$ for all $i$. 
The guarantee obtained by Theorem \ref{t5} is strictly better than the bound in Theorem \ref{t2} except for  the following corner case: 
$(i)$ $d_i^p=1$ i.e., there is a tie in some iteration $i$ between multiple users for the best resource, AND $(ii)$ every user who is part of the tie set has curvature equal to $ \underset{u \in \mathcal{U}}{\text{max }} c_u$ 
, in which case the bound matches with the bound of Theorem \ref{t2}, $\frac{1}{1+ \underset{u \in \mathcal{U}}{\text{max }} c_u}$.

\begin{lemma}\label{lem:partitiontight}The approximation guarantee obtained in Theorem \ref{t5} is tight.
\end{lemma}

\begin{bemark}\label{rem:intuition} One question that is important to understand is the exact form of dependence 
of the discriminant on the approximation guarantee for the greedy algorithm that emerges from Theorem \ref{t5}. Some intuition towards this end can be derived as follows. Recall that the total number of resources $|\cR|=n$.
Consider the instance where in the first $n-1$ iterations of the $\mathsf{GREEDY-M}$ algorithm, resources $r_1, \dots, r_{n-1}$ (indexed by choices of $\mathsf{GREEDY-M}$ algorithm) have been allocated to users $\hat{u}_1, \dots, \hat{u}_{n-1}$ that matches the allocation for the same resources by the optimal algorithm, i.e., $G^{n-1} = \Omega^{n-1}$ ($\Omega^{i}$ is the restriction of the optimal solution $\Omega$ to the items chosen in $G^{i}$). 
Then we claim that the one remaining resource is also allocated to the same user as done by the optimal algorithm, i.e., $G=\Omega$.
The reason is that the $\mathsf{GREEDY-M}$ algorithm in the last iteration allocates 
$r_n$ to the user $u_n$ that maximizes $Z(G^{n-1} \cup \{(u_n,r_n)\}) - Z(G^{n-1})$. 
Since the optimal allocation $\Omega$ maximizes $Z$, and $G^{n-1} = \Omega^{n-1}$, the $\mathsf{GREEDY-M}$ algorithm allocates $r_n$ to $\hat{u}_n$ (same as in the optimal solution).


Extending this scenario backwards, let $G^{n-2} = \Omega^{n-2}$, then $G=G^{n}$ may not be equal to $\Omega$, since the $\mathsf{GREEDY-M}$ algorithm in the final two iterations may make different choices compared to the optimal algorithm. In particular, let the $\mathsf{GREEDY-M}$ algorithm allocate resource $r_{n-1}$ to user $u_{n-1}$, while in the optimal algorithm $r_{n-1}$ is allocated to user ${\hat u}_{n-1}$, where  $u_{n-1} \ne {\hat u}_{n-1}$. Let the increment made by the $\mathsf{GREEDY-M}$ algorithm at iteration $n-1$ be 
$\rho_{n-1}$, while the increment made by the optimal algorithm in choosing user 
${\hat u}_{n-1}$ to allocate resource $r_{n-1}$ 
given its choices for the $n-2$ resources be $\rho^o_{n-1} = Z(\Omega^{n-2} \cup \{({\hat u}_{n-1},r_{n-1})\}) - Z(\Omega^{n-2})$, where $\Omega^{n-2} = G^{n-2}$. By definition of the $\mathsf{GREEDY-M}$ algorithm, $\rho^o_{n-1} \le \frac{\rho_{n-1}}{d_{n-1}^p}$.
Moving on to the final iteration, let the $\mathsf{GREEDY-M}$ algorithm accrue zero incremental valuation on allocating the last item (worst case) while the incremental valuation for the optimal algorithm on allocation of item $r_n$ to user ${\hat u}_{n}$ be $\rho^o_{n} = Z(\Omega^{n-1} \cup \{({\hat u}_{n},r_{n})\}) - Z(\Omega^{n-1})$. Note that $\rho^o_{n}\le \rho_{n-1}$, since otherwise the $\mathsf{GREEDY-M}$ algorithm would have allocated resource $r_{n}$ to user ${\hat u}_{n}$ in iteration $n-1$ itself.

Thus, the ratio of the $\mathsf{GREEDY-M}$ valuation to the optimal valuation is:
\begin{align*}
    \frac{Z(G)}{Z(\Omega)} &= \frac{Z(G^{n-2}) + \rho_{n-1}
    + 0}{Z(\Omega^{n-2}) + 
    \rho^o_{n-1}+ \rho^o_n}\ge  \frac{Z(\Omega^{n-2}) + \rho_{n-1} + 0}{Z(\Omega^{n-2}) + \frac{\rho_{n-1}}{d_{n-1}^p} + \rho_{n-1}} \ge \frac{1}{\frac{1}{d_{n-1}^p}+1}.
\end{align*}
Similar argument can be extended to earlier iterations of the $\mathsf{GREEDY-M}$ algorithm, to conclude that its approximation guarantee should depend on the discriminant $d_i^p$ as $1/d_i^p$.
\end{bemark}



In the next section, we consider the online version of the submodular partition problem, and show that the same approximation guarantee as derived in Theorem \ref{t5} (which now will be called competitive ratio) can be achieved by a natural online version of the $\mathsf{GREEDY-M}$ algorithm. 
\section{Online Monotone Submodular Partition Problem}

\begin{problem} \label{p3}
This problem is identical to Problem \ref{p2}, except that now, at each time index $t=1,2,\dots,|\cR|$, 
one resource $j_t \in \mathcal{R}, |\cR|=n$ arrives, which must immediately be allocated to exactly one of the users and the decision is irrevocable.
%
\end{problem}
For an online problem, given the arrival sequence of resources $\sigma$ that is a permutation over the order of arrival of $|\cR|=n$ resources, the competitive ratio of any online $A$ algorithm is defined as $\sfr_A = \min_{\sigma}\frac{Z(A(\sigma))}{Z(\Omega(\sigma))}$ and the objective is to find an optimal online algorithm $A^*$ such that $\sfr_{A^*} =  \max_A \sfr_A$. We propose a simple modification of the $\mathsf{GREEDY-M}$ algorithm to make it online and then bound its competitive ratio.

\begin{algorithm}[H] \label{alg2}
\caption{Modified Greedy Algorithm for the Online Monotone Submodular Partition problem}
\begin{algorithmic}[1]
\Procedure{$\mathsf{GREEDY-ON}$}{}
    \State \textbf{Initialize:} $G^0=\phi, t=1$
    \While{$t \le |\cR|=n$, on arrival of resource $j_t$ at time $t$}
    \State Allocate $j_t$ to user $u_t$ if 
        \State $u_t = \underset{u}{\text{ argmax}}\left\{ \rho_{j_t}^u(G^{t-1}) \right\}$ 
        \State \textbf{Tie:} Allocate resource $j_t$ to user $u$ with least $c_{u}$ 
        \State $G^t \gets G^{t-1} \cup \{(u_t,j_t)\}$
        \State $t \gets t+1$
    \EndWhile
    \State \textbf{Return} $G=G^n$
\EndProcedure
\end{algorithmic}
\end{algorithm}

\begin{bemark}
$\mathsf{GREEDY-ON}$ is simpler than $\mathsf{GREEDY-M}$ since at each iteration, the resource to be allocated is fixed, and the tie breaking rule does not require the discriminant information.
\end{bemark}

To derive a lower bound on the competitive ratio of the $\mathsf{GREEDY-ON}$ algorithm, we need the following definition of the discriminant.
\begin{definition} \label{defn:discon}
\noindent For iteration $t$ of $\mathsf{GREEDY-ON}$, where resource $j_t$ arrives and the current allocated set is $G^{t-1}$, the discriminant at iteration $t$, $d^o_t$ (where $^o$ stands for online) is defined as:
$$d^o_t = \frac{\rho_{j_t}^{u^*}(G^{t-1}) }{\underset{ u' \ne u^*}{\text{max }} \rho_{j_t}^{u'}(G^{t-1}) }, \qquad \text{where, } u^* = \underset{u}{\text{argmax}} \rho_{j_t}^u(G^{t-1}).$$
\end{definition}

\noindent The definition for discriminant in the online case only slightly differs from the offline case, in that the resource being allocated at time $t$ is fixed as $j_t$. 
Next, we present the final main result of the paper that provides a guarantee on the competitive ratio of the $\mathsf{GREEDY-ON}$ algorithm.

\begin{theorem} \label{t7}
For any arrival sequence $\sigma$ over the $|\cR|=n$ resources, the competitive ratio of the $\mathsf{GREEDY-ON}$ algorithm on Problem \ref{p3} is bounded by 
    $$\sfr_{\mathsf{GREEDY-ON}} \ge \text{min} \left(1,\frac{1}{\left(\underset{t}{\text{max }}\left\{\frac{1}{d^o_t} + c_{u_t}\right\}\right)}\right),$$
 where $d^o_t$ is the discriminant and $c_{u_t}$ is the curvature of the user chosen, in iteration $t$, respectively, for $t \in \{1,2,\dots,|\cR|\}$.
\end{theorem}

The natural online variant of the greedy algorithm $\mathsf{GREEDY}$ is known to have a  competitive ratio of at least $1/2$  (\cite{Fisher1978}), which can be improved to $1/(1+c)$ using the curvature information (\cite{CONFORTI1984251}). Theorem \ref{t7} shows that the competitive ratio can be further improved if the discriminant values for the problem instance are large.




\newpage
\section{Appendix}
Recall that given a set $S$ associated with an ordering $(s_1,s_2, \dots, s_{|S|})$, $S^{i}$ denotes the partial ordering $(s_1,s_2, \dots, s_{i})$.
\subsection{Intermediate Lemmas}

\begin{lemma} [\cite{CONFORTI1984251}, Equation 2.1] \label{l1}
    Given an instance of Problem $\ref{prob:1}$, for all 
    $ A,B \in \mathcal{M} $, the following is true:
    $$ Z(A \cup B) \le Z(A) + \sum_{i:b_i \in B \setminus A} \rho_{b_i}(A).$$
\end{lemma}
\begin{lemma} [\cite{CONFORTI1984251}, Lemma 2.2] \label{l2}
Given feasible solutions $A,B$ to Problem \ref{prob:1}, and an ordering $A = (a_1, a_2, \dots, a_K)$, an ordering for $B$ can be constructed as $(b_1, b_2, \dots, b_K)$ such that:
\begin{equation*}
    A^{i-1} \cup \{b_i\} \in \mathcal{M}, \quad \text{for } i = 1, 2, \dots, K.
\end{equation*}
Furthermore, if for any $i$, $b_i \in A \cap B$, then $b_i = a_i$.
\end{lemma}
%
%
%

\begin{bemark} \label{r1}
    The ordering for $B$ in Lemma \ref{l2} need not be unique. In a particular iteration of the construction, there may exist more than one $b_{t} : b_{t} \cup A^{t-1} \in \mathcal{M}$. The ordering exists irrespective of which $b_{t}$ is chosen.
\end{bemark}

\begin{corollary} \label{c1}
Given a feasible solution $A$ to Problem \ref{prob:1} with ordering $A=(a_1,a_2,\dots,a_K)$, if for some $i$, there exists a unique $a$ : $A^{i-1} \cup \{ a \} \in \mathcal{M}$ then, $a$ belongs to every feasible solution to Problem \ref{prob:1}.
\end{corollary}
\begin{proof}
Order any feasible solution $B$ to Problem \ref{prob:1} as per Lemma \ref{l2}. $A^{i-1} \cup b_{i} \in \mathcal{M} \Rightarrow b_i = a$, and hence $a \in B$.
\end{proof}

\begin{lemma} [\cite{CONFORTI1984251}, Lemma 2.2] \label{l3}
Given feasible solutions $A,B$ to Problem \ref{prob:1}, and an ordering $A = (a_1, a_2 \dots, a_K)$, ordering $B$ as per Lemma \ref{l2} guarantees that: \begin{align*} a_i \in A \setminus B &\Leftrightarrow b_i \in B \setminus A.
\end{align*}
\end{lemma}
We introduce Lemma \ref{l4} below as an aid for the proof in Section \ref{sec6p1} B.

\begin{lemma} \label{l4}
Given feasible solutions $A,B$ to Problem \ref{prob:1}, and an ordering $A = (a_1, a_2 \dots, a_K)$, ordering $B$ as per Lemma \ref{l2} also guarantees for all $i$ that:
\begin{align*}
&\exists A_c \subset A \setminus A^{i-1},\ |A_c| = K-i, \text{ such that: } A^{i-1} \cup \{b_i\} \cup A_c \in \mathcal{M}.
\end{align*}
\end{lemma}
\begin{proof}
    The statement is trivially true for $b_i \in A$ ($\Rightarrow b_i=a_i$), as we can choose $A_c = A \setminus (A^{i-1} \cup \{b_i \})$. For $b_i \not\in A$, the claim is proved as follows by showing that we can find a sequence of sets $\{ A^t_c \}$ such that:
    \begin{align} \label{eq1}
    A_c^t \subset A \setminus A^{i-1},\ |A_c| = t, \text{ such that: } A^{i-1} \cup \{b_i\} \cup A^t_c \in \mathcal{M}.
    \end{align}
    From the ordering of $B$ as per Lemma \ref{l2},  (\ref{eq1}) is trivial for $t=0$. For $t > 0$, we prove (\ref{eq1}) by induction:
    \begin{enumerate}
        \item Assume that (\ref{eq1}) is true for some $t$:
        $$ \Rightarrow \exists A_c^t \subset A \setminus A^{i-1} ,\ |A_c^t| = t, \text{ such that: } A^{i-1}\cup \{b_i\} \cup A_c^t \in \mathcal{M}.$$
        \item Consider the sets $A^{i-1}\cup \{b_i\} \cup A_c^t$ and $A$. By the augmentation property of matroids, \begin{equation} \label{eq2}
        \exists a \in A \setminus A^{i-1} \setminus A_c^t : A^{i-1}\cup \{b_i\} \cup A_c^t \cup \{a\} \in \mathcal{M}. \end{equation}
        \item We set $A^{t+1}_c \gets A^t_c \cup \{a \}\ [\subset A \setminus A^{i-1}]$. From (\ref{eq2}), we have $A^{i-1} \cup b_i \cup A^{t+1}_c \in \mathcal{M}$. Thus we have found an $A^{t+1}_c$, with $|A^{t+1}_c| = t+1$ that satisfies (\ref{eq1}). Plugging $t=K-i$ into (\ref{eq1}):
    \begin{align*}
        \Rightarrow\ \exists A_c \subset A \setminus A^{i-1},\ |A_c| = K-i, \text{ such that: } A^{i-1} \cup \{b_i\} \cup A_c \in \mathcal{M}.
    \end{align*}
    \end{enumerate}
\end{proof}

\begin{corollary} \label{c2}
(An extension to Corollary \ref{c1}) For some feasible solution $A$ to Problem \ref{prob:1}, if there exists a unique $S : S \subseteq N \setminus A^{i-1},\ |S| = K-i+1$ such that $A^{i-1} \cup S \in \mathcal{M}$, then $S$ is a part of every feasible solution to Problem \ref{prob:1}.
\end{corollary}
\begin{proof}
The set $A \setminus A^{i-1}$ is a valid candidate for $S$, since it has cardinality $K-i+1$ and $A^{i-1} \cup (A \setminus A^{i-1}) = A \in \mathcal{M}$. By the uniqueness of $S$, we conclude that $S = A \setminus A^{i-1} = (a_i,a_{i+1},\dots,a_K)$. Consider any feasible solution $B$ to Problem \ref{prob:1}. Order $B$ as per Lemma \ref{l2}. From Lemma \ref{l3}, we have that: 
    \begin{align} \label{eq3}
        &\exists A_c \subset A \setminus A^{i-1},\ |A_c| = K-i, \text{ such that:} A^{i-1} \cup \{b_i\} \cup A_c \in \mathcal{M}.
    \end{align}
Observe that: (a) $A_c \subset A \setminus A^{i-1}$, (b) $|A \setminus A^{i-1}| = K-i+1$, and (c) $|A_c| = K-i$.
Therefore $(A \setminus A^{i-1}) \setminus A_c$ is a singleton set containing some $a_{i_1}$, i.e. $A \setminus A^{i-1} = A_c \cup \{a_{i_1} \}$. Also note that: 
\begin{align} \label{eq4}
A = A_{i-1} \cup (A \setminus A^{i-1}) \in \mathcal{M},\ \Rightarrow\ A_{i-1} \cup (A_c \cup \{ a_{i_1} \}) \in \mathcal{M}.
\end{align}
Because of the uniqueness of $S : |S| = K-i+1$ and $S \cup A^{i-1} \in \mathcal{M}$, from (\ref{eq3}) and (\ref{eq4}), we have:
\begin{equation} \label{eq5}
a_{i_1} \in S  \quad \text{ and } \quad b_i = a_{i_1} \Rightarrow b_{i} \in A.
\end{equation}
By the ordering of $B$ as per Lemma \ref{l2}, $b_i \in A \Rightarrow b_i = a_i$. Thus, from (\ref{eq5}), we have $b_{i} = a_{i}$.
We also have $b_i = a_{i_1}$, which gives $a_i = a_{i_1}$. \vspace{\baselineskip}\newline
We can also argue that since $S = A \setminus A^{i-1}$ is the only set with $|S| = K-i+1$ such that $S \cup A^{i-1} \in \mathcal{M}$, $\forall t > i,\ S' = A \setminus A^{t-1}$ is the only set with $|S'| = K-t+1$ such that $S' \cup A^{t-1} \in \mathcal{M}$. Therefore, we can extend the arguments made in (\ref{eq3}), (\ref{eq4}) and (\ref{eq5}) to the set $A^{t-1}$ and conclude that $b_t = a_{i_t} = a_t$. Thus, $\{b_i,b_{i+1},\dots,b_K\} = \{a_i,a_{i+1},\dots,a_K \} = S \Rightarrow S \subseteq B$.
\end{proof}
\begin{lemma} \label{lemma:choices}
    There exists a unique $S : |S| = K-i+1$ such that $A^{i-1} \cup S \in \mathcal{M}$ iff there exist exactly $K-i+1$ choices for $a$ such that $A^{i-1} \cup \{a \} \in \mathcal{M}$.
\end{lemma}
    \begin{proof}
    The forward implication can be shown because there cannot be less than $K-i+1$ choices to add to $A^{i-1}$ (contradicts that $|S| = K-i+1$ as each element of $S$ is a valid choice) and there cannot be more than $K-i+1$ choices, because this contradicts the uniqueness of $S$ (using the augmentation property of matroids).\newline
    The reverse implication also directly follows from the augmentation property of matroids and from the rank of the matroid being equal to $K$.
    \end{proof}

\begin{lemma} [\cite{CONFORTI1984251} Lemma 2.1] \label{l5}
Given an instance of Problem \ref{prob:1}, for all $A,B \in \mathcal{M}$, with $A$ ordered as $(a_1,a_2, \dots, a_{|A|})$ and $A^t = (a_1,a_2, \dots, a_t)$, we have:
$$Z(B) \le \sum_{i:a_i \in A \cap B}\rho_{a_i}(A^{i-1}) + c \sum_{i:a_i\in A \setminus B} \rho_{a_i}(A^{i-1}) + \sum_{b \in B\setminus A}\rho_b(A).$$
\end{lemma}
\begin{lemma} \label{l6}
    Given an instance of Problem \ref{p2}, for all $A,B \in \mathcal{M}^p$, with $A$ ordered as $(a_1,a_2,\dots,a_{|A|})$ and $A^t=(a_1,a_2,\dots,a_t)$ we have:
\begin{equation*}
    Z(B) \le \sum_{\substack{i : a_i\in A \setminus B \\ a_i=(u_i,r_i)}} c_{u_i} \rho_{a_i}(A^{i-1}) + \sum_{i : a_i \in A \cap B} \rho_{a_i}(A^{i-1}) + \sum_{b \in B \setminus A} \rho_b (A),
\end{equation*}
where the curvature of $Z_{u_i}(S)$ is denoted by $c_{u_i}$.
\end{lemma}
\begin{proof}
\noindent Following the proof of Lemma \ref{l5}  [\cite{CONFORTI1984251} Lemma 2.1] , we can write:
    \begin{equation} \label{eq:36000}
    Z(B) \le \sum_{i:a_i \in A}\rho_{a_i}(A^{i-1}) - \sum_{\substack{i:a_i\in A \setminus B \\ a_i = (u_i,r_i)}} \rho_{a_i}(B \cup A^{i-1}) + \sum_{b \in B\setminus A}\rho_b(A).
		\end{equation}
		From the definition of curvature of $Z_u(S)$ in (\ref{def:curvature}), $\rho^u_r(B \cup A^{i-1}) \ge (1-c_u) \rho_r^u(A^{i-1})$. Thus, (\ref{eq:36000}) gives:
    \begin{align} \label{eq:46700}
		Z(B) \le \sum_{i:a_i \in A}\rho_{a_i}(A^{i-1}) - \sum_{\substack{i:a_i\in A \setminus B \\ a_i = (u_i,r_i)}} (1-c_{u_i}) \rho_{a_i}(A^{i-1}) + \sum_{b \in B\setminus A}\rho_b(A).
    \end{align}
		Rearranging terms in (\ref{eq:46700}) and using $\underset{i : a_i \in A}{\sum} \rho_{a_i}(A^{i-1}) - \underset{i : a_i \in A \setminus B}{\sum} \rho_{a_i}(A^{i-1}) = \underset{i : a_i \in A \cap B}{\sum} \rho_{a_i}(A^{i-1})$ gives:
		\begin{align*}
		Z(B) \le \sum_{i : a_i \in A \cap B} \rho_{a_i}(A^{i-1}) + \sum_{\substack{i : a_i\in A \setminus  B \\ a_i=(u_i,r_i)}} c_{u_i} \rho_{a_i}(A^{i-1}) + \sum_{b \in B \setminus  A} \rho_b (A).
\end{align*}
\end{proof}

\subsection{Proof of Theorem \ref{t6}}\label{app:proofgen}
For proving Theorem \ref{t6}, we restate the following definitions.
\begin{enumerate}
    \item The solution output by the $\mathsf{GREEDY}$ algorithm is denoted by $G =(g_1,g_2 \dots,\dots,  g_K)$ ordered in the sequence of elements chosen by the $\mathsf{GREEDY}$ algorithm, where $K$ is the rank of the matroid.
    \item Partial solutions output by $\mathsf{GREEDY}$: $G^t$ as $(g_1,g_2 \dots, g_t)$.
    \item $ \rho_i := \underset{g_i}{\text{max }} \rho_{g_i}(G^{i-1}) : g_i \cup G^{i-1} \in \mathcal{M} \quad$ [$\mathsf{GREEDY}$ increment at iteration $i$].
    \item Recalling $N$ as the ground set in Problem \ref{prob:1}, for any set $S \subseteq N$, $$S_\bot := \{ j : j \in N \setminus S,\ j \cup S \in \mathcal{M}\}.$$
\end{enumerate}
\begin{lemma} \label{l8} For any $g' \in G^{i-1}_\bot \setminus G$, where 
$G$ is the output of the $\mathsf{GREEDY}$ algorithm,
$$\rho_{g'}(G^{i-1}) \le \frac{\rho_{g_i}(G^{i-1})}{d_i} = \frac{\rho_i}{d_i}.$$
\end{lemma}
\begin{proof}
Since $g' \in G^{i-1}_\bot \setminus G$, we can write: \begin{equation} \label{eq18}
\rho_{g'}(G^{i-1}) \le \text{max}\left\{ \rho_{j'}(G^{i-1}) : j' \in G^{i-1}_\bot \setminus G \right\}.
\end{equation}
Observe that $G_\bot^{i-1} \setminus G \subseteq G_\bot^{i-1} \setminus \{ g_i \}$. Therefore, (\ref{eq18}) gives:
\begin{equation} \label{eq2098}
\rho_{g'}(G^{i-1}) \le \text{max}\left\{ \rho_{j'}(G^{i-1}) : j' \in G^{i-1}_\bot,\ j' \ne g_i \right\}.
\end{equation}
From the definition of $d_i$ in iteration $i$ (Definition \ref{defn:discgen}), we have:
\begin{align*}
&d_i = \frac{\rho_{g_i}(G^{i-1})}{\text{max}\left\{ \rho_{j'}(G^{i-1}) : j' \in G^{i-1}_\bot,\ j' \ne g_i \right\} },
\\
\text{which implies that} \ \  &\text{max}\left\{ \rho_{j'}(G^{i-1}) : j' \in G^{i-1}_\bot,\ j' \ne g_i \right\} = \frac{\rho_{g_i}(G^{i-1})}{d_i}.
\end{align*}
Together with (\ref{eq2098}), this completes the proof.
\end{proof}

\begin{lemma} \label{l9}
Given that $G$ is ordered as $(g_1,g_2,\dots,g_K)$, ordering $\Omega$ as $(\omega_1,\omega_2 \dots, \omega_K)$ as per Lemma \ref{l2} guarantees for any $1 \le t \le K$:
\begin{equation*}
\rho_{\omega_i}(G^{t-1}) \le 
\begin{cases}
\rho_{t} &\quad \forall i: i > t, \\ 
\frac{\rho_i}{d_i} &\quad \forall i: 1 \le i \le t, \\
\end{cases}
\quad i: \omega_i \in \Omega\setminus G.
\end{equation*}
\end{lemma}
\begin{proof} Recall that $i_0$ is defined as $\min \{i: |G_\bot^{i-1}| = K-i+1 \}$. Recall from the discussion in Section \ref{sec6p1} B, for $i \ge i_0,\ g_i \in G \cap \Omega$. Therefore the set $\{i : \omega_i \in \Omega \setminus G\}$ must be a subset of $\{i:i < i_0\}$. As long as $i < i_0$, $d_i$ is always defined (even if it is defined as $\infty$). Given that $G$ is ordered as $(g_1,g_2,\dots,g_K)$, order $\Omega$ as $(\omega_1,\omega_2 \dots, \omega_K)$ as in Lemma \ref{l2}. Consider $\rho_{\omega_i}(G^{t-1})$ for the $2$ cases:
\begin{enumerate}
\item If $i > t$, by the ordering of $\Omega$, $G^{i-1} \cup \omega_i \in \mathcal{M}$. Since, $G^{t-1} \cup \{\omega_i\} \subseteq G^{i-1} \cup \{\omega_i\}$, by the independence system property of matroids $G^{t-1} \cup \{\omega_i\} \in \mathcal{M}$. Thus, from the solution of $\mathsf{GREEDY}$, $G$, we have: $\rho_{\omega_i}(G^{t-1}) \le \rho_{g_{t}}(G^{t-1}) = \rho_{t}$ for $i > t$.
\item If $1 \le i \le t$, by submodularity, $\rho_{\omega_i}(G^{t-1}) \le \rho_{\omega_i}(G^{i-1}) \quad -(a)$.\newline Some facts about $\omega_i$:
    \begin{enumerate}
    \item[$(b)$] $G^{i-1} \cup \{\omega_i \} \in \mathcal{M} \Rightarrow \omega_i \in G^{i-1}_\bot$ [due to the ordering of $\Omega$],
    \item[$(c)$] $\omega_i \in \Omega \setminus G \Rightarrow \omega_i \not\in G$.
    \item[$(d)$] From $(b)$ and $(c)$, we have $\omega_i \in G^{i-1}_\bot \setminus G$. From Lemma \ref{l8}, we have $\rho_{\omega_i}(G^{i-1}) \le \frac{\rho_i}{d_i}$. 
    \end{enumerate}
    Together with $(a)$, for $1 \le i \le t$, this gives: $\rho_{\omega_i}(G^{t-1}) \le \frac{\rho_i}{d_i}.$
\end{enumerate}
\end{proof}
Now we are ready to prove Theorem \ref{t6}.
\begin{proof}[Proof of Theorem \ref{t6}]
From Lemma \ref{l5} for $\Omega$ and $G$, we have:
\begin{equation*}
    Z(\Omega) \le c \sum_{i : g_i\in G \setminus  \Omega} \rho_{g_i}(G^{i-1}) + \sum_{i : g_i \in G \cap \Omega} \rho_{g_i}(G^{i-1}) + \sum_{\omega \in \Omega \setminus  G} \rho_\omega (G).
\end{equation*}

\noindent Following the ordering for $\Omega$ in Lemma \ref{l2}, it follows that:
\begin{equation} \label{eq19}
    Z(\Omega) \le c \sum_{i : g_i\in G \setminus  \Omega } \rho_i + \sum_{i : g_i \in G \cap \Omega} \rho_i + \sum_{i: \omega_i \in \Omega \setminus  G} \rho_i \frac{\rho_{\omega_i}(G^{i-1})}{\rho_i}.
\end{equation}
\noindent From Lemma \ref{l3}, under the ordering of $\Omega$ from Lemma \ref{l2},
\begin{align} \label{eq3454}
    g_i\in G \setminus \Omega \Leftrightarrow \omega_i \in \Omega \setminus G.
\end{align}
We define
\begin{equation} \label{eq788}
    d_i' := \frac{\rho_i}{\rho_{\omega_i}(G^{i-1})}.
\end{equation}
For $i : \omega_i \in \Omega \setminus G$, we can lower bound $d_i'$ by $d_i$. This is a consequence of Lemma \ref{l9}, with $t=i$:
\begin{align}
\rho_{\omega_i}(G^{i-1}) &\le \frac{\rho_{i}}{d_i} \nonumber \\
\Rightarrow\ d_i &\le d_i' \label{eq:dom}
\end{align}

\noindent Consider the last term in the RHS of (\ref{eq19}):
\begin{align}
    \sum_{i: \omega_i \in \Omega \setminus  G} \rho_i \frac{\rho_{\omega_i}(G^{i-1})}{\rho_i}
    &\overset{(i)}{=} \sum_{i: \omega_i \in \Omega \setminus G} \frac{\rho_i}{d_i'} \overset{(ii)}{=} \sum_{i: g_i \in G \setminus \Omega} \frac{\rho_i}{d_i'}, \label{eq556}
\end{align}
where $(i)$ follows from the definition of $d_i'$ in (\ref{eq788}) and $(ii)$ follows from the assertion in (\ref{eq3454}). Substituting (\ref{eq556}) back in (\ref{eq19}), we have:
\begin{align} \label{eq20}
    Z(\Omega) &\le c \sum_{i : g_i\in G \setminus  \Omega } \rho_i + \sum_{i : g_i \in G \cap \Omega} \rho_i + \sum_{i: g_i \in G \setminus \Omega} \frac{\rho_i}{d_i'}, \nonumber \\
    &= \sum_{i : g_i\in G \setminus  \Omega } \left(c + \frac{1}{d_i'}\right)\rho_i + \sum_{i : g_i \in G \cap \Omega} \rho_i, \nonumber \\
    &\le \left(c +\underset{i : g_i \in G \setminus \Omega}{\text{max }} \ \frac{1}{d_i'}\right) \sum_{i : g_i\in G \setminus  \Omega } \rho_i + \sum_{i : g_i \in G \cap \Omega} \rho_i.
\end{align}

\noindent We reduce (\ref{eq20}) as below:
\begin{align} \label{eq21}
    Z(\Omega) \le \text{max}\left(1,\ c +\underset{i : g_i \in G \setminus \Omega}{\text{max }} \ \frac{1}{d_i'}\right) \sum_{i : g_i\in G} \rho_i = \text{max}\left(1,\ c +\underset{i : g_i \in G \setminus \Omega}{\text{max }} \ \frac{1}{d_i'}\right) Z(G).
\end{align}

\noindent Since we cannot compute $G \setminus \Omega$ in polynomial time, we relax the domain over which $\frac{1}{d_i'}$ is maximized. Recall from the discussion in Section \ref{sec6p1} B: $\forall i \ge i_0,\ g_i \in G \cap \Omega$. Therefore $\{i:i \ge i_0 \} \subseteq \{ i: g_i \in G \cap \Omega \}$, which means that $\{i:i < i_0 \} \supseteq \{ i: g_i \in G \setminus \Omega \}$.\newline
Therefore, in the term involving maximization of $\frac{1}{d_i'}$ over $G \setminus \Omega$ in (\ref{eq20}), we replace $ G \setminus \Omega$ by the larger set $\{i: i < i_0 \}$ and perform the maximization over this set. This gives:
\begin{align} \label{eq22}
    Z(\Omega) \le \text{max}\left(1,\ c +\underset{i < i_0}{\text{max }} \ \frac{1}{d_i'}\right) Z(G).
\end{align}
From \eqref{eq:dom}, lower bounding $d_i'$ by $d_i$ in (\ref{eq22}) gives:
\begin{align*}
    Z(\Omega) \le \text{max}\left(1,\ c +\underset{i < i_0}{\text{max }} \ \frac{1}{d_i}\right) Z(G) &\le \text{max}\left(1,\ c + \ \frac{1}{\underset{i < i_0}{\text{min }}d_i}\right) Z(G),\\
    &\overset{(i)}{=} \text{max}\left(1,c+\frac{1}{d_{\min}} \right) Z(G),
\end{align*}
where $(i)$ follows from the definition of $d_{\min}$ (Definition \ref{defn:discgen}). Hence we get, 
\begin{equation*}
    \frac{Z(G)}{Z(\Omega)} \ge\frac{1}{ \text{max}\left(1,\frac{1}{d_{\min}}+c\right)} = \text{min} \left( 1, \frac{1}{ c + \frac{1}{d_{\min}}} \right) = \text{min} \left( 1, \frac{1}{\left( c + \underset{i < i_0}{\text{max }} \ \frac{1}{d_i}\right)} \right).
\end{equation*}

\end{proof}
\subsection{Proof of Theorem \ref{t5}}
Before proving Theorem \ref{t5}, we restate the following definitions.
 In the following, we denote the increment  made by the $\mathsf{GREEDY-M}$ algorithm in iteration $i$ by choosing the user-resource pair $g_i = (u_i,r_i)$, as 
$ \rho_i =  \rho^{u_i}_{r_i}(G^{i-1}) := \rho_{g_i}(G^{i-1})$ and the set chosen by the $\mathsf{GREEDY-M}$ algorithm as $G$.
\begin{lemma} \label{l7}
With $|\cR|=n$, let the solution generated by the $\mathsf{GREEDY-M}$ algorithm $G$ be ordered as $(g_1,g_2,  \dots, g_n)$ (in order of resources allocated by the $\mathsf{GREEDY-M}$ algorithm), where $g_i = (u_i,r_i)$ and the partial greedy solutions $G^t=(g_1,g_2, \dots,g_t)$ [with $G^0 = \phi$]. Ordering $\Omega$ as $(\omega_1,\omega_2,\dots,\omega_n)$, where $\omega_i = (\hat{u}_i,r_i)$ [arranged in the same order of resources $r_i^{{\text 's}}$ allocated by the $\mathsf{GREEDY-M}$ algorithm] gives:
\begin{equation*}
\rho^{\hat{u}_i}_{r_i}(G^t) \le 
\begin{cases}
\rho_{t+1} &\forall i: i > t,\\ 
\frac{\rho_i}{d_i^p} \quad &\forall i: i \le t,
\end{cases} \qquad
i: (\hat{u}_{i},r_i) \in \Omega \setminus G,
\end{equation*}
where $\rho_{r_i}^{\hat{u}_i}(G^t)$ follows the same notation as in (\ref{eq:rhodef}) and is equal to $Z(G^t \cup \{(\hat{u}_i,r_i)\}) - Z(G^t)$.
\end{lemma}
\begin{proof}$ $
Consider an index $i : (\hat{u}_i,r_i) \in \Omega\setminus G \subseteq \Omega \setminus G^t$. If no such $i$ exists, we are guaranteed that $\Omega \setminus G = \phi \Rightarrow \Omega = G$. Otherwise, for such $i$, consider $\rho_{r_i}^{\hat{u}_i}(G^t)$:
\begin{enumerate}
\item For $i > t$, $G^t \cup (\hat{u}_i,r_i) \in \mathcal{M}^p$, since the resource $r_i$ has not yet been allocated in $G^t$. Thus for $i>t$, $\rho_{t+1} \ge \rho_{r_i}^{\hat{u}_i}(G^t)$.
\item $\mathsf{GREEDY-M}$ has allocated the resources $\{r_1,r_2, \dots, r_t\}$ to $G^t$ until iteration $t$, implying that $(u_i,r_i) \in G^t$ for any $i \le t$. Therefore, defining $S_{t,i} = G^t \cup \{(\hat{u}_i,r_i)\}$, we have that for $i \le t$, $S_{t,i} \not\in \mathcal{M}^p$. This is because $\{(u_i,r_i),(\hat{u}_i,r_i)\} \subset S_{t,i}$, and hence it fails to satisfy $|S_{t,i} \cap \mathcal{V}_{r_i}| \le 1$ - a condition every set in $\mathcal{M}^p$ must satisfy as defined in Problem \ref{p2}. However,
\begin{enumerate}
    \item[(a)] $G^{i-1} \cup \{(\hat{u}_i,r_i)\} \in \mathcal{M}^p$ (since $r_i$ has not yet been allocated in $G^{i-1}$), and
    \item[(b)] $\hat{u}_i \ne u_i$. This is because:
    \begin{align*}
        (\hat{u}_i,r_i) &\in \Omega \setminus G, \\ \Rightarrow (\hat{u}_i,r_i) &\not\in G = \{(u_1,r_1),(u_2,r_2),\dots,(u_n,r_n)\}, \\
        \Rightarrow (\hat{u}_i,r_i) &\ne (u_i,r_i). \qquad [\text{Since } (u_i,r_i) \text{ is an element of } G]
    \end{align*}
\end{enumerate}
\noindent Thus, from (a) and (b) we have that $(\hat{u}_i,r_i)$ is a valid choice for $\mathsf{GREEDY-M}$ in the $i^{th}$ iteration, but it involves allocating $r_i$ to a user different from the greedy choice in that iteration. Therefore, we conclude that $\hat{u}_i$ is at most the second best user choice for $r_i$ in iteration $i$ and we get the inequality $(i)$ that for $i\le t$
\begin{align}
\rho_{r_i}^{\hat{u}_i} (G^{i-1}) \overset{(i)}{\le} \underset{u \ne u_i}{\text{max }} \rho^{u}_{r_i} (G^{i-1}) \overset{(ii)}{\le} \frac{1}{d_i^p}\rho_{r_i}^{u_i}(G^{i-1}) = \frac{\rho_i}{d_i^p}. \label{eq1232}
\end{align}
\noindent where $(ii)$ follows from the definition of the discriminant at iteration $i$. It is important to note that $(i)$ is conditionally true only if $(\hat{u}_i,r_i) \in \Omega \setminus G$. Since $r_i$ has not yet been allocated in $G^{i-1}$, we are assured that $G^{i-1} \cup \{(\hat{u}_i ,r_i)\} \in \mathcal{M}^p$. Thus, from (\ref{eq1232}) we have:
$$ \text{for } i\le t,\ \ \rho^{\hat{u}_i}_{r_i}(G^t) \overset{(iii)}{\le} \rho_{r_i}^{\hat{u}_i} (G^{i-1}) \le \frac{\rho_i}{d_i^p}$$
where $(iii)$ follows from submodularity.
\end{enumerate}
\end{proof}
 
\subsubsection{Proof of Theorem \ref{t5}}

\begin{proof}
\noindent From Lemma \ref{l6}, we have:
\begin{equation} \label{eq14}
    Z(\Omega) \le \sum_{\substack{i : g_i\in G \setminus  \Omega \\ g_i=(u_i,r_i)}} c_{u_i} \rho_{g_i}(G^{i-1}) + \sum_{i : g_i \in G \cap \Omega} \rho_{g_i}(G^{i-1}) + \sum_{\omega \in \Omega \setminus  G} \rho_\omega (G).
\end{equation}

\noindent Let $G$ be ordered as $(g_1,g_2,\dots,g_n)$, where $g_i = (u_i,r_i)$. Recall that,
$$g_i = \underset{u,r}{\text{ argmax}}\left\{ \rho_r^u(G^{i-1}) : \ r \in \mathcal{R}(G^{i-1}_\bot) \right\}.$$
We order $\Omega$ as $(\omega_1,\omega_2,\dots,\omega_n)$, where $\omega_i = (\hat{u}_i,r_i)$ [in the same order of resources $r_i^{{\text 's}}$ allocated by $\mathsf{GREEDY-M}$]. Consider the last summation in the RHS of (\ref{eq14}):
\begin{align} \label{eq15}
    \sum_{\omega \in \Omega \setminus G} \rho_\omega(G) = &\sum_{i: \omega_i \in \Omega \setminus G} \rho_{\omega_i}(G)
    \overset{(i)}{\le} \sum_{i: \omega_i \in \Omega \setminus G} \rho_{\omega_i}(G^{i-1})
    \overset{(ii)}{\le} \sum_{i: \omega_i \in \Omega \setminus G} \frac{\rho_i}{d_i^p},
\end{align}
where $(i)$ follows from submodularity of $Z(S)$ and $(ii)$ follows from Lemma \ref{l7}. 
Putting together (\ref{eq14}) and (\ref{eq15}), we have:
\begin{align} \label{eq16}
    Z(\Omega) &\le \sum_{\substack{i : g_i\in G \setminus  \Omega \\ g_i=(u_i,r_i)}} c_{u_i} \rho_i + \sum_{i : g_i \in G \cap \Omega} \rho_i + \sum_{i: \omega_i \in \Omega \setminus  G} \frac{\rho_i}{d_i^p}.
\end{align}

\noindent In order to combine terms in the first and the third summations of (\ref{eq16}), we give an argument to show that if $\omega_i = (\hat{u}_i,r_i) \in \Omega \setminus G$ then $g_i = (u_i,r_i) \in G \setminus \Omega$ for the same index $i$. Consider $i$ such that $\omega_i = (\hat{u}_i,r_i) \in \Omega \setminus G$. If no such $i$ exists, $\Omega \setminus G = \phi \Rightarrow \Omega = G$, and hence the statement of Theorem \ref{t5} holds true trivially. Otherwise, it follows that:
\begin{enumerate}
\item[(a)] $\omega_i \not\in G$.
\item[(b)] Since $ g_i = (u_i, r_i) \in G$, for such $i$, it follows from (a) that $\omega_i \ne g_i$. This means that $\hat{u}_i \ne u_i$.
\item[(c)] From (b) it follows that if $\hat{u}_i \ne u_i$, then $(u_i,r_i) \not\in \Omega$ (since $\Omega \in \mathcal{M}^p$, it cannot contain both $(\hat{u}_i,r_i)$ and $(u_i,r_i)$). Hence $ g_i \in G \setminus \Omega$.
\end{enumerate}
Since $|G \setminus \Omega| = |\Omega \setminus G|$ (which is a consequence of $|G| = |\Omega| = n$), we combine the first and the third summations in (\ref{eq16}) to give:
\begin{align} \label{eq17}
    Z(\Omega) \le \sum_{\substack{i : g_i\in G \setminus \Omega \\ g_i = (u_i,r_i)}} \left(c_{u_i}+\frac{1}{d_i^p} \right) \rho_i + \sum_{i : g_  i \in G \cap \Omega} \rho_i.
\end{align}
Since, $
    1 \le \text{max}\left(1,\frac{1}{d_i^p}+c_{u_i}\right), \text{ and } \left(c_{u_i}+\frac{1}{d_i^p} \right) \le  \text{max}\left(1,\frac{1}{d_i^p}+c_{u_i}\right)$ from (\ref{eq17}), we have:
\begin{align*}
    Z(\Omega) &\le
    \sum_{i = 1}^{n} \text{max}\left(1,\frac{1}{d_i^p}+c_{u_i}\right) \rho_i \le \underset{i}{\text{max}}\left(\text{max}\left(1,\frac{1}{d_i^p}+c_{u_i}\right) \right) Z(G).
\end{align*}

\noindent Thus, $   \frac{Z(\Omega)}{Z(G)} \le\frac{1}{ \underset{i}{\text{max}}\left(\text{max}\left(1,\frac{1}{d_i^p}+c_{u_i}\right) \right)} = \text{min} \left( 1, \frac{1}{ \underset{i}{\text{max}}\left(c_{u_i} + \frac{1}{d_i^p}\right)}\right)$.
\end{proof}

\subsection{Proof of Theorem \ref{t7} - Online Partition Problem}
\noindent Before proceeding ahead, we restate the following notation:
\begin{enumerate}
\item User chosen by the $\mathsf{GREEDY-ON}$ algorithm at time $t$: $u_t$.
\item The $\mathsf{GREEDY-ON}$ solution, $G = \{(u_1,j_1),(u_2,j_2),\dots,(u_n,j_n)\} = \{g_1,g_2,\dots,g_n \}$ and optimal solution, $\Omega = \{(\hat{u}_1,j_1),(\hat{u}_2,j_2),\dots, (\hat{u}_n,j_n)\} = \{\omega_1,\omega_2,\dots,\omega_n \}$.
\item Partition matroid $\mathcal{M}^p$ as defined in Problem $2$. Then, $\sum_{u \in \mathcal{U}} Z_u(S_u) \ \equiv \ Z(S) : S \in \mathcal{M}^p$.
\item The set of all user-resource pairs involving the resource $j_t$: $U^t = \{(u,j_t): \forall u \in \mathcal{U}\}$
\item $\Psi^t = \Omega \cap U^t$ and $\xi^t = G \cap U^t$.
\end{enumerate}

\begin{bemark} \label{r3}
    $| \Psi^t \cap \xi^t | \le 1$. This is because at every time $t$ only a single resource is allocated ($|\Psi^t| = |\xi^t| = 1$).
\end{bemark}

\begin{lemma} \label{lemma:intersection}
    For $t_1 \ne t_2$, $\Psi^{t_1} \cap \Psi^{t_2} = \xi^{t_1} \cap \xi^{t_2} = \Psi^{t_1} \cap \xi^{t_2} = \phi$.
\end{lemma}
\begin{proof}
It follows from definition that $\Psi^t \subseteq U^t$ and $\xi^t \subseteq U^t$. Since, $U^{t_1} \cap U^{t_2} = \phi$, the result follows.
\end{proof}

\begin{lemma} \label{l10}
$G \setminus \Omega$ and $\Omega \setminus G$ can be decomposed as:
\begin{align*}
G \setminus \Omega = \bigcup_{t=1}^n \xi^t \setminus \Psi^t,\ \text{ and }\
\Omega \setminus G = \bigcup_{t=1}^n \Psi^t \setminus \xi^t.
\end{align*}
\end{lemma}
\begin{proof}
From the definition of $\Psi^t$ and $\xi^t$, it follows that: $G = \bigcup_{t=1}^n \xi^t, \ \text{ and }\
\Omega = \bigcup_{t=1}^n \Psi^t.$ The proof concludes using Lemma \ref{lemma:intersection}.
\end{proof}

\subsubsection{Proof of Theorem \ref{t7}}

\begin{proof}
\noindent From Lemma \ref{l1}, we have:
\begin{align} \label{eq23}
Z(\Omega \cup G) & \le Z(G) + \sum_{q \in \Omega \setminus  G} \rho_q(G), \nonumber \\
&\overset{(i)}{=} Z(G) + \sum_{t=1}^{n} \sum_{q \in \Psi^t \setminus  \xi^t} \rho_q(G), \nonumber \\
&\overset{(ii)}{\le} Z(G) + \sum_{t=1}^n \sum_{q \in \Psi^t\setminus \xi^t} \rho_q(G^{t-1}),
\end{align}
where $(i)$ follows from Lemma \ref{l10} and Remark \ref{lemma:intersection}, and $(ii)$ follows from submodularity. Expanding the LHS of (\ref{eq23}), and again applying Lemma \ref{l10} and Remark \ref{lemma:intersection} gives:

\begin{align} \label{eq24}
Z(\Omega \cup G) = Z(\Omega) + \sum_{t=1}^n \sum_{q \in \xi^t \setminus  \Psi^t} \rho_{q}(\Omega \cup G^{t-1}).
\end{align}

\noindent Combining (\ref{eq23}) and (\ref{eq24}):
\begin{align} \label{eq25}
    Z(\Omega) &\le Z(G) + \sum_{t=1}^n \sum_{q \in \Psi^t\setminus \xi^t} \rho_q(G^{t-1}) - \sum_{t=1}^n \sum_{q \in \xi^t \setminus  \Psi^t} \rho_{q}(\Omega \cap G^{t-1}).
\end{align}
From the definition of curvature, for any $q = (u_t,j_t)$:
\begin{equation} \label{eq26}
    \rho_{q}(\Omega \cap G^{t-1}) \ge (1-c_{u_t}) \rho_q(G^{t-1}).
\end{equation}
Expanding $Z(G)$ in (\ref{eq25}) and combining it together with (\ref{eq26}), we have:
\begin{align}
    Z(\Omega) &\le \sum_{t=1}^n \sum_{q \in \xi^t} \rho_q(G^{t-1}) + \sum_{t=1}^n \sum_{q \in \Psi^t\setminus \xi^t} \rho_q(G^{t-1}) - \sum_{t=1}^n \sum_{\substack{q \in \xi^t \setminus \Psi^t \\ q = (u_t,j_t)}} (1-c_{u_t})\rho_{q}(G^{t-1}). \label{eq4342}
\end{align}
Cancelling identical terms in the first and last summations in (\ref{eq4342}), we get:
\begin{align} \label{eq27}
    Z(\Omega) &\le \sum_{t=1}^n \sum_{q \in \xi^t \cap \Psi^t} \rho_q(G^{t-1}) + \sum_{t=1}^n \sum_{q \in \Psi^t\setminus \xi^t} \rho_q(G^{t-1}) + \sum_{t=1}^n \sum_{\substack{q \in \xi^t \setminus \Psi^t \\ q = (u_t,j_t)}} c_{u_t}.\rho_{q}(G^{t-1}).
\end{align}

\noindent Consider the middle summation of (\ref{eq27}). Recall that $\hat{u}_t$ is the user that the resource $j_t$ is allocated to in the optimal solution $\Omega$. For $q = (\hat{u}_t,j_t) \in \Psi^t\setminus \xi^t$:
\begin{enumerate}
    \item $(\hat{u}_t,j_t) \not\in \xi^t$. But, $\xi_t$ contains $(u_t,j_t)$ (since it is equal to $G \cap U^t$), and hence we conclude that $\hat{u}_t \ne u_t$.
    \item We also note that $G^{t-1} \cup (\gamma_{n-1}, j_t) \in \mathcal{M}^p$, $\forall \gamma_{n-1} \in \mathcal{U}$ (since $j_t$ has not been allocated in $G^{t-1}$), and hence $G^{t-1} \cup (\hat{u}_t,j_t) \in \mathcal{M}^p$.
\end{enumerate}
Therefore, from the definition of discriminant (Definition \ref{defn:discon}) for the online case, we can write for $q \in \Psi^t\setminus \xi^t$:
\begin{equation} \label{eq879}
\rho_q(G^{t-1}) \le \frac{\rho_{j_t}^{u_t}(G^{t-1})}{d_t^o}.
\end{equation}
Recalling, $\rho_t := \rho_{j_t}^{u_t}(G^{t-1})$ and putting together (\ref{eq27}) with (\ref{eq879}), we have:
\begin{align} \label{eq28}
    Z(\Omega) &\le \sum_{t=1}^n \left( \sum_{q \in \xi^t \cap \Psi^t} \rho_t + \sum_{q \in \Psi^t\setminus \xi^t} \frac{\rho_t}{d_t^o} + \sum_{\substack{q \in \xi^t\setminus \Psi^t \\ q = (u_t,j_t)}} c_{u_t}.\rho_t \right).
\end{align}
Recalling that $\Psi^t = \{ (\hat{u}_t,j_t)\}$ and $\xi^t =\{ (u_t,j_t)\}$. If $\xi^t \setminus \Psi^t \ne \phi$, then it is clear that $\Psi^t \setminus \xi^t \ne \phi$. Thus, we combine together the last two summations in (\ref{eq28}) to get 
\begin{align*}
    Z(\Omega) &\le \sum_{t=1}^n \sum_{q \in \xi^t \cap \Psi^t} \rho_t + \sum_{t=1}^n\sum_{\substack{q \in \xi^t\setminus \Psi^t \\ q = (u_t,j_t)}} \left( c_{u_t} + \frac{1}{d_t^o} \right)\rho_t.
\end{align*}
Bringing $c_{u_t} + \frac{1}{d_t^o}$ out of the summation by defining $\Lambda = \underset{t}{\text{max}}\left(c_{u_t}  +\frac{1}{d_t^o}\right)$, we get
\begin{align} \label{eq29}
    Z(\Omega) \le \sum_{t=1}^n \sum_{q \in \xi^t \cap \Psi^t} \rho_t + \Lambda \sum_{t=1}^n\sum_{\substack{q \in \xi^t\setminus \Psi^t \\ q = (u_t,j_t)}} \rho_t.
\end{align}
We can further reduce the RHS of (\ref{eq29}) to:
\begin{align*} 
    Z(\Omega) &\le \text{max}\left(1,\Lambda\right) \sum_{t=1}^n \sum_{q \in \xi^t \cap \Psi^t} \rho_t + \text{max}\left(1,\Lambda\right) \sum_{t=1}^n\sum_{\substack{q \in \xi^t\setminus \Psi^t \\ q = (u_t,j_t)}} \rho_t, \\
    &= \text{max}\left(1, \Lambda\right)\sum_{t=1}^n \sum_{q \in \xi^t} \rho_t \ = \ \text{max}\left(1, \Lambda\right) Z(G).
\end{align*}
This gives the required bound:
\begin{align*}
    \frac{Z(G)}{Z(\Omega)} \ge \text{min}\left(1, \frac{1}{\Lambda} \right) = \text{min}\left(1, \frac{1}{\underset{t}{\text{max}}\left(c_{u_t}  +\frac{1}{d_t^o}\right)}\right).
\end{align*}
\end{proof}


\subsection{Proof of Lemma \ref{lem:partitiontight} - Tight example for the bound in Theorem \ref{t5}:}\label{app:tightpart}
\begin{enumerate}
\item Set of users, $\mathcal{U} = \{ u_1, u_2 \}$, and set of resources, $\mathcal{R} = \{ r_1,r_2,\dots,r_{n}\}$.
\item For some $0 \le c \le 1$, $1 \le d \le \frac{1}{1-c}$, and some small $\epsilon>0$, let $d_-=d-\epsilon$.
\item Define the monotone submodular valuation functions, $Z_{u_1}(S)$ and $Z_{u_2}(S)$:\vspace{0.4\baselineskip}\newline
\begin{tabular}{|c|c|c|c|c|c}
    \hline & $r_1$ & $r_2$ & $r_3$ & $r_4$ & $\cdots$ \\
    \hline $u_1$ & $d$ & \begin{tabular}{|c|c|}
        $r_1 \not\in S_{u_1}$ & $r_1 \in S_{u_1}$ \\
        $d_-$ & $d_-(1-c)$
    \end{tabular} & $d^3(1-c)^2$
    & \begin{tabular}{|c|c|}
        $r_3 \not\in S_{u_1}$ & $r_3 \in S_{u_1}$ \\
        $d_-^3(1-c)^2$ & $d_-^3(1-c)^3$
    \end{tabular}
    & $\cdots$ \\
    \hline $u_2$ & $1$ & $d^2(1-c)$ & \begin{tabular}{|c|c|}
        $r_2 \not\in S_{u_2}$ & $r_2 \in S_{u_2}$ \\
        $d_-^2(1-c)$ & $d_-^2(1-c)^2$
    \end{tabular} & $d^4(1-c)^3$ & $\cdots$ \\ \hline
\end{tabular}\vspace{0.4\baselineskip}\newline

Subsequently,

\begin{center}
\begin{tabular}{|c|c|c|c|}
    \hline & $r_i : i$ is odd & $r_i : i$ is even \\
    \hline $u_1$ & $d^i(1-c)^{i-1}$
    & \begin{tabular}{|c|c|}
        $r_{i-1} \not\in S_{u_1}$ & $r_{i-1} \in S_{u_1}$ \\
        $d_-^{i-1}(1-c)^{i-2}$ & $d_-^{i-1}(1-c)^{i-1}$
    \end{tabular} \\
    \hline $u_2$ & \begin{tabular}{|c|c|}
        $r_{i-1} \not\in S_{u_2}$ & $r_{i-1} \in S_{u_2}$ \\
        $d_-^{i-1}(1-c)^{i-2}$ & $d_-^{i-1}(1-c)^{i-1}$
    \end{tabular} & $d^{i}(1-c)^{i-1}$ \\ \hline
\end{tabular}
\end{center}
\item $Z_{u_1}(S)$ and $Z_{u_2}(S)$ both have curvature $c$. This is discussed in more detail in Appendix \ref{app:int1}.
\item $\mathsf{GREEDY-M}$ \textbf{solution:}
\begin{enumerate}
    \item $(u_1,r_1)$ has the maximum initial valuation. Thus, the $\mathsf{GREEDY-M}$ algorithm allocates $r_1$ to $u_1$. The discriminant for the $1^{st}$ iteration is $d_1 = d$.
    \item $(u_2,r_2)$ has the maximum greedy increment for the $2^{nd}$ iteration. The discriminant for the $2^{nd}$ iteration is $d_2 = \frac{d^2}{d_-}$.
    \item Continuing as above, it follows that the greedy solution generated by $\mathsf{GREEDY-M}$ is $G = \{ (u_1,r_1),(u_2,r_2),(u_1,r_3),(u_2,r_4),\dots, \}$. Hence,
    \begin{align}
    Z(G) &= d(1+d(1-c) + (d(1-c))^2 + \dots + (d(1-c))^{n-1}) \nonumber \\ &= \frac{d(1-(d(1-c))^n)}{1-d(1-c)}. \label{eq4560}
    \end{align}
    \item[$\bullet$] In this problem, $c_{u_1} = c_{u_2} = c$ as shown in Appendix \ref{app:int1}, meaning that $\forall i,\ c_{u_i} = c$. For iteration $1$, the discriminant, $d_1 = d$.
    Step $2$ onwards, the discriminant $d_i=\frac{d^2}{d_-}$. Since $d_- < d$, for this problem, from Theorem $5$: $$\frac{Z(G)}{Z(\Omega)} \ge \frac{1}{c+\frac{1}{d}}.$$
\end{enumerate}
\item The optimal solution is $\Omega = \{ (u_2,r_1),(u_1,r_2),(u_2,r_3),(u_1,r_4),\dots, \}$, shown in Appendix \ref{app:opttightex}. Hence, 
\begin{align}
    Z(\Omega) &= 1+d_-(1+d_-(1-c) + (d_-(1-c))^2 + \dots + (d_-(1-c))^{n-1}) \nonumber \\
    &= 1+\frac{d_-(1-(d_-(1-c))^{n})}{1-d_-(1-c)}. \label{eq334}
\end{align}

\item Thus, from (\ref{eq4560}) and (\ref{eq334}), $$ \lim_{\substack{\epsilon \to 0, \\ n \to \infty}} \frac{Z(G)}{Z(\Omega)} = \lim_{\epsilon \to 0} \frac{d}{1-d_-(1-c) + d_-} = \lim_{\epsilon \to 0} \frac{d}{1+d_- c} = \frac{1}{\frac{1}{d}+c}.$$
\end{enumerate}
Thus, from (\ref{eq4560}) it follows that the guarantee for Theorem \ref{t5} is met.

\subsubsection{Curvatures of \texorpdfstring{$Z_{u_1}(S)$}{Zu1S} and \texorpdfstring{$Z_{u_2}(S)$}{Zu2S}:}\label{app:int1}

The definition of curvature for $Z_{u_1}(S)$ and $Z_{u_2}(S)$ is:
\begin{equation} \label{eq2200}
    \text{for }i = 1,2 \quad c_{u_i} = 1 - \underset{S,r \in S^*}{\text{min }} \frac{\rho_r^{u_i}(S)}{\rho_r^{u_i}(\phi)}, \qquad \text{where } S^* = \{ r: r \in \mathcal{R} \setminus S,\ \rho_r^{u_i}(\phi) > 0\}.
\end{equation}
For any $r \in \mathcal{R}$, the minimizer for $S$ in (\ref{eq2200}) is $ \mathcal{R} \setminus \{r \}$. This is because for some $r$, any valid choice for $S$ cannot contain $r$ (since $r \in S^*$), and hence $S \subseteq \mathcal{R} \setminus \{r_i \}$. Submodularity of $Z_{u_1}(S)$ guarantees that the increment upon adding $r$ to any such $S$ would be more than that for $ \mathcal{R} \setminus \{r_i \}$. Considering $S = \mathcal{R} \setminus \{r \}$ itself is a valid choice, since $S^* = \{ r \}$, which proves that it is the minimizer for $S$.

\noindent For $u_1$, the expression for curvature is evaluated below:
\begin{equation} \label{eq4567}
\frac{\rho_r^{u_1}(S)}{\rho_r^{u_1}(\phi)} =
\begin{cases}
\frac{d_-^{i-1}(1-c)^{i-1}}{d_-^{i-1}(1-c)^{i-2}} = 1- c & r_i : i \text{ is even, and } S = \mathcal{R} \setminus \{r_i\}, \\
\frac{d^i(1-c)^{i-1} + d_-^{i-1}(1-c)^{i} - d_-^{i-1}(1-c)^{i-1}}{d^i(1-c)^{i-2}} \overset{(\dagger)}{>} 1-c \quad & r_i : i \text{ is odd, and } S = \mathcal{R} \setminus \{ r_i\},
\end{cases}
\end{equation}
where $(\dagger)$ follows from the fact that $d_- < d$. It follows from (\ref{eq4567}) and the definition of curvature in (\ref{eq2200}) that:
$$ c_{u_1} = 1 - (1-c) = c.$$

\noindent A similar analysis for $u_2$ gives:
\begin{equation} \label{eq45678}
\frac{\rho_r^{u_2}(S)}{\rho_r^{u_2}(\phi)} =
\begin{cases}
\frac{d_-^{i-1}(1-c)^{i-1}}{d_-^{i-1}(1-c)^{i-2}} = 1- c & r_i : i \text{ is odd, and } S = \mathcal{R} \setminus \{r_i\}, \\
\frac{d^i(1-c)^{i-1} + d_-^{i-1}(1-c)^{i} - d_-^{i-1}(1-c)^{i-1}}{d^i(1-c)^{i-2}} > 1-c \quad & r_i : i \text{ is even, and } S = \mathcal{R} \setminus \{ r_i\}.
\end{cases}
\end{equation}
Together with the definition of curvature in (\ref{eq2200}), this once again gives:
$$ c_{u_2} = 1 - (1-c) = c.$$

\subsubsection{Optimal solution for tight example in Theorem \ref{t5}}\label{app:opttightex}

The optimal solution for the tight example in Appendix \ref{app:tightpart} is $\Omega = \{(u_2,r_1),(u_1,r_2),(u_2,r_3),(u_1,r_4),\dots, \}$
The reasoning behind this is presented below. From the structure of the problem instance in Appendix \ref{app:tightpart}, it follows that any valid solution $\Omega'$ can be equivalently represented as an ordered set of users $(\alpha_{1},\alpha_{2},\dots,\alpha_{K})$, where $\alpha_{i}$ is either $u_1$ or $u_2$ and represents the user to which $r_i$ is allocated, and $\Omega'$ is ordered as per $r_1,\dots,r_K$, the order in which $\mathsf{GREEDY-M}$ allocates resources. Any such $\Omega'$ can be converted to $\Omega$ by replacing every $\alpha_i : i$ is odd, by $u_2$ and every $\alpha_i : i$ is even, by $u_1$. We sequentially update $\Omega'$ in the following manner until is not further possible:

\begin{enumerate}
    \item Consider $i$ such that $\alpha_i = \alpha_{i+1} = u_1$. From the function valuations of $Z_{u_1}(S)$ and $Z_{u_2}(S)$ and since $d \le \frac{1}{1-c}$, it follows that for odd $i$, replacing $\alpha_i$ by $u_2$ and for even $i$, replacing $\alpha_{i+1}$ by $u_2$ would increase the valuation of $\Omega'$. $\Omega'$ is updated by performing this exchange.

    \item Consider all $i$ such that $\alpha_i = \alpha_{i+1} = u_2$. From the function valuations of $Z_{u_1}(S)$ and $Z_{u_2}(S)$ and since $d \le \frac{1}{1-c}$, it follows that for odd $i$, replacing $\alpha_{i+1}$ by $u_1$ and for even $i$, replacing $\alpha_i$ by $u_1$ would increase the valuation of $\Omega'$. $\Omega'$ is updated by performing this exchange.
\end{enumerate}

\noindent $\Omega'$ is recursively updated as per iterations $1$ and $2$, with each iteration increasing the valuation of $\Omega'$ until there no longer exists $i$ such that $\alpha_i = \alpha_{i+1} = u_1$ or $\alpha_{i} = \alpha_{i+1} = u_2$. We claim that unless $\Omega'$ is initially $\{u_1,u_2,u_1,\dots, \}$, after the final update $\Omega'$ must be equal to $\{u_2,u_1,u_2,\dots, \}$. This is because, as per the update rule, for all $i$, $\alpha_i$ and $\alpha_{i+1}$ cannot both be $u_1$ or $u_2$. Hence the final $\Omega'$ is either $\{u_2,u_1,u_2,\dots, \}$ or $\{u_1,u_2,u_1,\dots, \}$. If even a single update has taken place, the update rule ensures that either some $\alpha_i : i$ is odd is replaced by $u_2$ (or) some $\alpha_i : i$ is even is replaced by $u_1$. This eliminates $\{u_1,u_2,u_1,\dots, \}$ as the final outcome. If no update has taken place, it follows that $\Omega' = \{u_1,u_2,u_1,\dots, \}$ since it is the only set other than $\{u_2,u_1,u_2,\dots, \}$ that cannot be updated.

\noindent Thus, we can transform any solution $\Omega'$ to $\Omega$ with a net increase in valuation. This proves that $\Omega$ is the optimal solution.

\subsection{Proof of Lemma \ref{lem:submodtight} - Tight example for the bound in Theorem \ref{t6}}\label{app:tightgen}
We show that using the $\mathsf{GREEDY}$ algorithm on the following problem instance, matches the bound in (\ref{eq22}):
\begin{enumerate}
    \item Consider the ground set $N$ as: $\{ \nu_1,\nu_2,\dots,\nu_K, \epsilon_1, \epsilon_2, \dots, \epsilon_{K} \}$.
    \item Define the matroid as $(N,\mathcal{M})$, where $\mathcal{M} = \{ S: S \subseteq N, S \cap \{ \nu_i, \epsilon_i\} \le 1,\ \forall i\}$ [only one among $\nu_i$ and $\epsilon_i$ can be members of any set $S \in \mathcal{M}$].
    \item For some $d \ge 1$, define the valuation function as:
    \begin{equation*} 
    \begin{split}
            Z(S) = d\sum_{I_1} (d(1-c))^{i-1} + \sum_{i \in I_2} (d(1-c))^{i-1} \} + d \sum_{\ \in I_3} (d(1-c))^{i-2}.
    \end{split}
    \end{equation*}
    where $I_1 = \{i :\epsilon_i \in S \}$, $I_2 = \{i : \nu_i \in S \text{ and } \left[ \epsilon_{i-1} \in S \text{ or } i=1 \right] \}$ and $I_3 = \{i : \nu_i \in S \text{ and } \epsilon_{i-1} \not\in S \}$.
    
    The idea behind such a valuation function is that at every iteration $i$ of the $\mathsf{GREEDY}$ algorithm, there is a tie between $\epsilon_i$ and $\nu_{i+1}$ (as seen below), and $\mathsf{GREEDY}$ picks arbitrarily among the two, say $\epsilon_i$ which is the sub-optimal choice.
    \item $Z(S)$ can also be represented in the form of its marginal increments as:
    \begin{align}
    &\rho_{\nu_i}(S) = \begin{cases}
        d(d(1-c))^{i-2}, &\epsilon_{i-1} \not\in S, \\
        (d(1-c))^{i-1}, &\epsilon_{i-1} \in S \text{ or } i=1,
    \end{cases} \nonumber\\
    &\rho_{\epsilon_i}(S) = \begin{cases}
        d(d(1-c))^{i-1}, &\nu_{i+1} \not\in S, \\
        (d(1-c))^{i}, &\nu_{i+1} \in S.
    \end{cases}
    \label{vals}
    \end{align}
    \item Before proceeding with the analysis for the greedy algorithm, observe that the optimal solution is $\Omega = \{ \nu_1, \nu_2, \dots, \nu_{K} \}$.
    
 This is shown below by showing a sequence of iterations by which any valid solution $\Omega'$ can be converted to $\Omega$ with an increase in valuation in every iteration. From the structure of the matroid $(N, \mathcal{M})$, it follows that any valid solution $\Omega'$ can be ordered as $(\alpha_1,\alpha_2,\dots,\alpha_K )$ (where $\alpha_i$ is either $\nu_i$ or $\epsilon_i$), and can be converted to $\Omega$ by replacing every $\epsilon_i \in \Omega'$ by $\nu_i$. Noting this point, a sequence of iterations converting any such $\Omega'$ to $\Omega$ with an increase in valuation in every iteration is:
    \begin{enumerate}
        \item[i.] If $\epsilon_K \in \Omega'$, update $\Omega'$ by replacing $\epsilon_K$ by $\nu_K$. From the incremental valuations defined in (\ref{vals}), it follows that such an exchange provides a non-negative increment $\sigma_1$ to the valuation of $\Omega'$. In case $\nu_K \in \Omega'$, set $\sigma_1 = 0$.
        \item[ii.] Starting from $K$, iterate over $i$ in the reverse order sequentially updating $\Omega'$; for every $i$ such that $\epsilon_{i}$ precedes $\nu_{i+1}$, update $\Omega'$ by replacing $\epsilon_{i}$ by $\nu_{i}$. In case $\nu_i \in \Omega'$, no update is made to $\Omega'$ in that iteration. From the function increments in (\ref{vals}), in every iteration, this iteration gives a positive increment $\sigma_i$ ($\sigma_i$ is set as $0$ in case there is no exchange).
        \item[$\bullet$] In the $i^{th}$ iteration of this process, $\Omega'$ would be $(\alpha_1,\alpha_2,\dots,\alpha_{K-i},\nu_{K-i+1},\dots,\nu_{K})$.
    \end{enumerate}
    Thus, the above series of steps sequentially convert any $\Omega'$ to $\Omega$, and give:
    \begin{equation*}
    Z(\Omega) = Z(\Omega') + \sum_{i=1}^K \sigma_i, \qquad \forall i, \sigma_i \ge 0.
    \end{equation*}
    This shows for every valid solution $\Omega'$, $Z(\Omega) \ge Z(\Omega').$
    Thus, we conclude that $\Omega$ is the optimal solution, having valuation:
    \begin{equation} \label{eq3300}
        Z(\Omega) = 1 + d\sum_{i=1}^K (d(1-c))^{i-1}.
    \end{equation}
    \item In order to compute $d_i'$ for the bound in (\ref{eq22}), we need to compute the ordering of $\Omega$ as per Lemma \ref{l2}. If the greedy solution is chosen in the order $G = (\alpha_{i_1},\alpha_{i_2},\dots,\alpha_{i_K})$ (where $\alpha_{i_t}$ takes either $\epsilon_{i_t}$ or $\nu_{i_t}$), an ordering for $\Omega$ as per Lemma \ref{l2} is just $(\nu_{i_1},\nu_{i_2},\dots,\nu_{i_K})$. It follows from the structure of the matroid, that for all $t$, the set $S_t = G^{t-1} \cup \{ \nu_t \} = \{ \alpha_{i_1},\alpha_{i_2} \dots,\alpha_{i_{t-1}}\} \cup \{\nu_{i_t} \}$ does not contain both $\nu_{i_j}$ and $\epsilon_{i_j}$ (since $\alpha_{i_j}$ is uniquely either $\epsilon_{i_j}$ or $\nu_{i_j}$) for all $j$. Hence $\forall t,\ S_t \in \mathcal{M}$. This shows that $(\nu_{i_1},\nu_{i_2},\dots,\nu_{i_K})$ is a valid ordering as per Lemma \ref{l2}.

    \item We now compute the $\mathsf{GREEDY}$ solution:
    \begin{enumerate}
        \item \textbf{iteration $\mathbf{1}$}: $\epsilon_1$ and $\nu_2$ are tied ($d_1^g=1$), both giving increment, $\rho_{\epsilon_1}(\phi) = \rho_{\nu_2}(\phi) = d$. Assume that the $\mathsf{GREEDY}$ algorithm arbitrarily picks $\epsilon_1$ between the two. We have:
        $$d_1' \overset{(i)}{=} \frac{\rho_{\epsilon_1}(\phi)}{\rho_{\nu_1}(\phi)} \overset{(ii)}{=} \frac{d}{1} = d,$$
        where $(i)$ follows from the definition of $d_i'$ in (\ref{eq788}), and from the ordering of $\Omega$ established previously in point 6, and $(ii)$ follows from the incremental valuations defined in (\ref{vals}).
        \item \textbf{iterations $\mathbf{i = 2, \dots, K}$}: $\epsilon_i$ and $\nu_{i+1}$ are tied ($d_i=1$), both giving increment $\rho_{\epsilon_i}(G^{i-1}) = \rho_{\nu_{i+1}}(G^{i-1}) = d(d(1-c))^{i-1}$. Assume that the greedy algorithm picks $\epsilon_i$ arbitrarily, which gives:
        \begin{equation} \label{eq800}
        d_i' \overset{(iii)}{=}  \frac{\rho_{\epsilon_i}(G^{i-1})}{\rho_{\nu_i}(G^{i-1})} \overset{(iv)}{=} \frac{d(d(1-c))^{i-1}}{d(1-c)^{i-1}} = d,
        \end{equation}
        where $(iii)$ once again follow from the definition of $d_i'$ in (\ref{eq788}), and from the ordering of $\Omega$ established in point 6, and $(iv)$ follows from the function increments in (\ref{vals}).
    \end{enumerate}
    Thus, the greedy solution is $G = \{\epsilon_1,\epsilon_2,\dots,\epsilon_K\}$, and has valuation $Z(G) = d \sum_{i=1}^K (d(1-c))^{i-1}.$\newline
    Therefore, with the optimal valuation in (\ref{eq3300}) we have: 
    \begin{align}
    \lim_{K \to \infty} \frac{Z(G)}{Z(\Omega)} &= \lim_{K \to \infty} \frac{d \sum_{i=1}^K (d(1-c))^{i-1}}{1+d \sum_{i=1}^K (d(1-c))^{i-1}},\nonumber\\ &= \frac{d}{1-d(1-c)+d},\nonumber\\ &= \frac{1}{\frac{1}{d}+c}. \label{eq7700}
    \end{align}
    \item To compute the bound in (\ref{eq22}), we need to compute $ \underset{i<i_0}{\text{max }} \frac{1}{d_i'}$. Observe that in the $i^{th}$ stage of the greedy algorithm, we have $2(K-i+1) > K-i+1$ choices, $\{ \nu_i,\nu_{i+1},\dots,\nu_K, \epsilon_i, \epsilon_{i+1}, \dots, \epsilon_{K} \}$. Therefore $i_0 = K+1$. Putting this together with (\ref{eq800}), we have:
    \begin{equation*}
        \underset{i<i_0}{\text{max }} \frac{1}{ d_i'} = \underset{i\le K}{\text{max }} \frac{1}{ d_i'} = \frac{1}{d}.
    \end{equation*}
    
    \item Therefore the bound in (\ref{eq22}) reduces to:
    \begin{equation*}
        \frac{Z(G)}{Z(\Omega)} \ge \frac{1}{\frac{1}{d}+c},
    \end{equation*}
    which matches with the limit ($K \rightarrow \infty$) in (\ref{eq7700}).
\end{enumerate}